\documentclass[11pt]{article}
\usepackage[utf8]{inputenc}
\usepackage[margin=2.5cm]{geometry}
\usepackage{amsmath,amsthm,amssymb,amsfonts}
\usepackage{hyperref}
\usepackage[OT1]{fontenc}
\usepackage{xcolor}
\usepackage{tcolorbox}
\usepackage{setspace}
\usepackage{float}
\usepackage{tikz}
\usepackage{adjustbox}
\usepackage{multirow}
\usetikzlibrary{decorations.pathreplacing}
\PassOptionsToPackage{colorlinks,citecolor=blue,urlcolor=blue}{hyperref}
\usepackage{graphicx}
\usepackage{bm}
\usepackage{subcaption}
\usepackage{hyperref}
\usepackage[backend=biber,style=numeric,sorting=ynt]{biblatex}
\definecolor{amethyst}{rgb}{0.6, 0.4, 0.8}

\newcommand{\E}{\mathbb{E}}
\newcommand{\argmax}{\mathrm{argmax}}
\newcommand{\M}{\mathcal{M}}
\newcommand{\T}{\mathcal{T}}
\renewcommand{\Pr}{\mathrm{Pr}}
\newcommand{\eps}{\epsilon}

\newcommand{\swdiff}{\mathrm{SWDIFF}}
\newcommand{\fair}{\mathrm{FAIR}}
\newcommand{\loss}{\mathrm{LOSS}}

\newcommand{\Y}{\mathcal{Y}}
\renewcommand{\subset}{\subseteq}
\newcommand{\MoLong}{\mathtt{DPExpMed}_\alpha}
\newcommand{\MoLongs}{\mathtt{DPExpMed}_{\alpha^*}}
\newcommand{\Mo}{\mathcal{M}^{med}_\alpha}

\newcommand{\ctm}{\mathcal{CTM}}
\newcommand{\spm}{\mathcal{SPM}}

\newcounter{case}[section]

\newtheoremstyle{boldthm}
  {\topsep}  
  {\topsep}   
  {\normalfont} 
  {}   
  {\bfseries} 
  {.}    
  { }  
  {}  

\theoremstyle{boldthm}

 \newtheorem{theorem}{Theorem}[section] 
\newtheorem{proposition}[theorem]{Proposition} 
\newtheorem{corollary}[theorem]{Corollary}
\newtheorem{procedure}[theorem]{Procedure}
\newtheorem{definition}[theorem]{Definition}
\newtheorem{lemma}[theorem]{Lemma}

\newtheoremstyle{remarkstyle}
  {3pt}
  {3pt}
  {\normalfont}
  {}
  {\itshape}
  {.}
  {0.5em}
  {\thmname{#1}\thmnumber{ #2}\thmnote{ (#3)}}

\theoremstyle{remarkstyle}

\theoremstyle{definition}

\title{Tradeoffs in Privacy, Welfare, and Fairness for Facility Location\thanks{This paper is based on Tang’s undergraduate thesis~\cite{Tang2025}, advised by Fish, Gonczarowski, and Vadhan. It will appear in the Symposium on Foundations of Responsible Computing, 2026.\vspace{0.8em}}}
\author{Sara Fish \thanks{School of Engineering and Applied Sciences, Harvard University. Email: sfish@g.harvard.edu. Fish was supported by an NSF Graduate Research Fellowship and a Kempner Institute Graduate Fellowship.}\and 
Yannai A. Gonczarowski \thanks{Department of Economics and Department of Computer Science, Harvard University. Email: yannai@gonch.name. Gonczarowski gratefully acknowledges research support by the National Science Foundation (NSF-BSF grant No.\ 2343922) and Harvard FAS Inequality in America Initiative.}\and 
Jason Z. Tang \thanks{Department of Computer Science, Harvard University. Email: jasontang@alumni.harvard.edu.}\and 
Salil Vadhan\thanks{School of Engineering and Applied Sciences, Harvard University. Email: salil\_vadhan@harvard.edu. Supported in part by NSF grant BCS-2218803.}}
\date{April 11, 2026}
\begin{document}

\maketitle

\thispagestyle{empty}
\begin{abstract}
    The differentially private (DP) facility location problem seeks to determine a socially optimal placement for a public facility while ensuring that each participating agent’s location remains private. In order to privatize its input data, a DP mechanism must inject noise into its output distribution, producing a placement that will have lower expected social welfare than the optimal spot for the facility. The privacy-induced welfare loss can be viewed as the ``cost of privacy,'' illustrating a tradeoff between social welfare and privacy that has been the focus of prior work. Yet, the imposition of privacy also induces a third consideration that has not been similarly studied: \emph{fairness} in how the ``cost of privacy'' is distributed across individuals. For instance, a mechanism may satisfy differential privacy with minimal social welfare loss, yet still be undesirable if that loss falls entirely on one individual. In this paper, we quantify this new notion of unfairness and design mechanisms for facility location that attempt to simultaneously optimize across these three objectives of privacy, social welfare, and fairness.

 Under this setup, we first derive an impossibility result, showing that privacy and fairness cannot be simultaneously guaranteed over all possible datasets that could represent the locations of individuals in a population. We then consider a relaxation of the original problem that still requires worst-case differential privacy, but only seeks fairness and appealing social welfare over smaller, more ``realistic-looking'' families of datasets. For this relaxation, we construct a DP mechanism and demonstrate that it is simultaneously optimal (or, for a harder family of datasets, near-optimal up to small factors) on fairness and social welfare. This suggests that while there is a tradeoff between privacy and each of social welfare and fairness, there is no additional tradeoff when we consider all three objectives simultaneously, provided that the population data is sufficiently natural.
\end{abstract}

\newpage

\setcounter{page}{1}
\pagenumbering{arabic} 
\pagestyle{plain}  
\section{Introduction}

The \emph{facility location} problem in game theory and economics examines the scenario where a public facility should be placed to maximize the total welfare of its prospective users. While this social welfare objective is the primary focus of traditional literature on facility location~\cite{Hot1929, DW1978, FSY2016, OCL2020}, differential privacy~\cite{Dwo2006, DMN2006} has recently become an additional important desideratum for applications that rely on personal data. Since the welfare-maximizing location for the facility may depend on sensitive information about the population (e.g., which members of the population would use the facility and where they reside), the facility location mechanism should not inadvertently leak its data inputs.

A facility location mechanism can only satisfy differential privacy (DP) by reducing its dependence on individualized data, which leads to lower social welfare compared to the optimal (but non-private) mechanism. While previous work in DP facility location attempts to minimize the \emph{total} loss in welfare~\cite{GLM2009, JNN2020, FLL2023}, this objective fails to consider the unfairness induced by the \emph{distribution} of this welfare loss across participating agents.

For example, consider the facility location problem depicted by Figure~\ref{fig:facility_location}, where an agent's welfare is inversely related to their distance from the facility. Upon switching from the optimal facility placement to a potential placement chosen under differential privacy, the \emph{total} distance between agents and the facility only increases moderately. However, Agent 4 bears almost all of that increase in distance while the remaining agents are minimally impacted. If, over the randomness of the DP mechanism, the changes in distance are frequently uneven across agents, then the mechanism might be undesirable as it introduces a new form of unfairness caused directly by the enforcement of privacy.

\begin{figure}[h]
    \centering
    \includegraphics[width=1.0\textwidth]{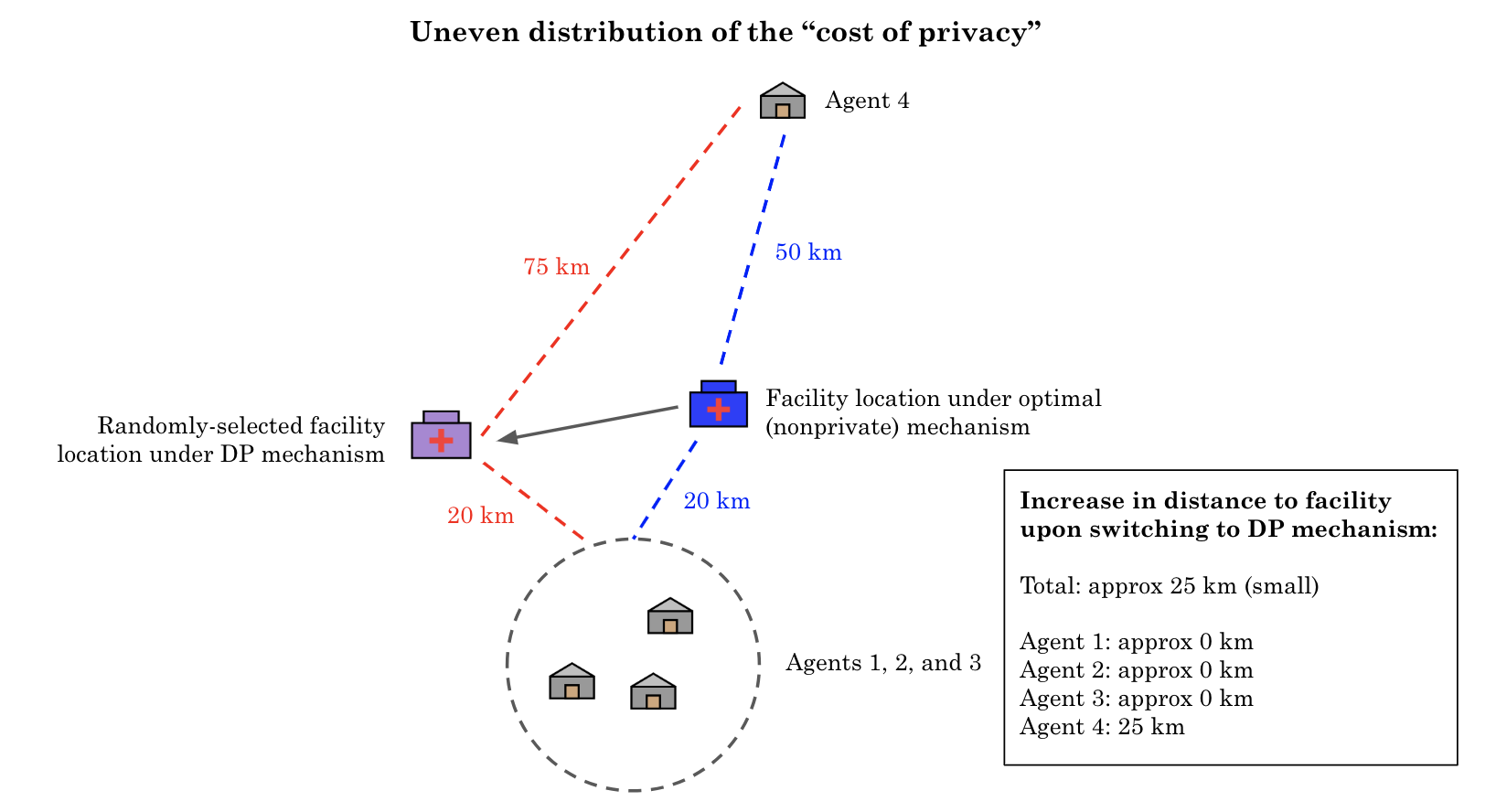}
    \caption{A depiction of how the switch from the optimal mechanism to the differentially private one can disparately impact agent utilities. In this case, the change in distance to the facility differs significantly between agents, suggesting a new form of unfairness.}
    \label{fig:facility_location}
\end{figure}

In this paper, we consider three objectives in facility location mechanism design simultaneously, by characterizing the tradeoff between privacy and its cost in terms of both social welfare and fairness. We require privacy through DP and introduce two metrics (or loss functions) to quantify the remaining objectives of social welfare and fairness. Both metrics are formulated by comparing the utilities of agents in two worlds: (1) the world where the location is selected via the optimal (but nonprivate) facility location mechanism $\T$, and (2) the world where the facility location is selected via a DP mechanism $\M$.
\begin{itemize}
    \item The first loss function, $\swdiff$, corresponds to social welfare. $\swdiff$ measures the difference in social welfare under the outputs of $\T$ and $\M$ and is standard in the literature.
    \item The second loss function, $\fair$, is our novel proposal and measures the \emph{maximum} loss in utility experienced by an individual agent when we move from $\T$ to $\M$. $\fair$ corresponds to fairness in the sense that minimizing it ensures no individual pays significantly for the imposition of privacy (in the spirit of Rawlsian social welfare in economics, which measures the minimum utility of any individual agent).
\end{itemize}

We analyze the facility location problem under these metrics. We focus specifically on the bounded and one-dimensional setting, under which the non-private optimum for $\swdiff$ is given by the median. We offer the following contributions.
\begin{enumerate}
    \item \textbf{Quantifying fairness under the imposition of privacy:} We propose a new individual-utility-based notion of fairness, $\fair$, as described above. This metric measures the extent to which privacy-induced welfare loss may be distributed unevenly among agents.
    \item \textbf{Impossibility result on achieving both privacy and fairness:} We show that no meaningfully private (i.e., $\eps$-DP for any reasonably small $\eps$) facility location mechanism preserves fairness, as quantified by $\fair$, for all possible arrangements (henceforth referred to as ``datasets'') of agents.
    \item \textbf{Simultaneous privacy, social welfare, and fairness under a relaxation:} The previous impossibility result only arises upon considering datasets that are ``pathological'' and unrepresentative of many real-world ones. We thus consider a relaxation of the problem where we still require mechanisms to be (pure) DP over all datasets, but only require good fairness and welfare over a smaller family of natural datasets. We propose and motivate two such families---datasets that are ``collapsing towards the median'' and datasets that are roughly ``single-peaked at the median''---on which we demonstrate that there exists a DP mechanism that achieves optimal (or, for the latter harder family, near-optimal) bounds on $\swdiff$ and $\fair$ \emph{simultaneously}. To arrive at this result, we establish three sets of technical results for each family:
    \begin{enumerate}
        \item \textbf{Analyzing a strong DP mechanism:} We first show that there exists a DP mechanism, $\MoLong$, that performs well on both $\swdiff$ and $\fair$ upon tuning its parameter value $\alpha$. The mechanism is an instantiation of the ``widened exponential mechanism''~\cite{McS2009} for the median.
        \item \textbf{Deriving high probability upper bounds:} We identify an optimal parameter value $\alpha=\alpha^*$ and prove upper bounds on $\swdiff$ and $\fair$ for the mechanism $\MoLongs$ over each family of datasets.

        \item \textbf{Deriving information-theoretic lower bounds:} We prove information-theoretic lower bounds for $\swdiff$ and $\fair$ that must be incurred by \emph{any} DP mechanism over each family of datasets.

    \end{enumerate}

    The key results from parts (b) and (c) are detailed in Table~\ref{tab:our_contributions}. All pairs of upper and lower bounds either match or nearly match up to small factors (see Table~\ref{tab:our_contributions} caption), suggesting our proposed mechanism from (a) is simultaneously optimal on both $\swdiff$ and $\fair$. This then implies that social welfare and fairness can co-exist when designing optimal private mechanisms for facility location. 
\begin{table}[H]
    \centering
    \small 
    \renewcommand{\arraystretch}{1.3}
    \begin{adjustbox}{max width=\textwidth}
    \begin{tabular}{|c|c||c|c|}
    \hline
    \textbf{Dataset family} &  & \textbf{Lower bound} & \textbf{Upper bound achieved by $\MoLongs$} \\
    \hline\hline
    \multirow{2}{*}{\textbf{Smaller family $\ctm$}} 
    & $\fair$ 
    & $\Omega\left(\frac{m}{n\epsilon}\ln \frac{1}{\beta}\right)$ 
    & $O\left(\frac{m}{n\epsilon}\ln \frac{1}{\beta}\right)$ \\
    \cline{2-4}
    & $\swdiff$ 
    & $\Omega\left(\frac{m}{n\epsilon^2} \left(\ln \frac{1}{\beta} \right)^2 \right)$ 
    & $O\left(\frac{m}{n\epsilon^2} \left(\ln \frac{1}{\beta} \right)^2 \right)$ \\
    \hline
    \multirow{2}{*}{\textbf{Larger family $\spm_\lambda$}} 
    & $\fair$ 
    & $\Omega\left(\frac{m}{n\epsilon}\ln \frac{1}{\beta} + m\lambda\right)$ 
    & $O\left(\frac{m}{n\epsilon} \ln \frac{n\lambda}{\beta} + m\lambda\right)$ \\
    \cline{2-4}
    & $\swdiff$ 
    & $\Omega\left(\frac{m}{n\epsilon^2} \left(\ln \frac{1}{\beta} \right)^2 + m\lambda\right)$ 
    & $O\left(\frac{m}{n\epsilon^2} \left(\ln \frac{n\lambda}{\beta}\right)^2 + \frac{m\lambda}{\epsilon}\right)$ \\
    \hline
    \multirow{2}{*}{\textbf{General datasets}} 
    & $\fair$ 
    & $m/2$ 
    & -- \\
    \cline{2-4}
    & $\swdiff$ 
    & -- 
    & -- \\
    \hline
    \end{tabular}
    \end{adjustbox}
    \caption{Lower and upper bounds on loss derived in Sections~\ref{upper_bounds} and \ref{lower_bounds}. $n$ is the number of agents, $m$ is the maximum possible distance between agents, $\beta$ is the DP failure probability, and $\lambda$ is a parameter that is typically $O(1/\sqrt{n})$. Our proposed $\epsilon$-DP mechanism achieves all upper bounds simultaneously, and the lower bounds apply to all $\epsilon$-DP mechanisms. The bounds match closely, with only minor discrepancies. Under typical assumptions $\lambda = O(1/\sqrt{n})$ and $\epsilon = \Theta(1)$, the extra $1/\epsilon$ in the $\spm_\lambda$ upper bound on $\swdiff$ disappears, and the extra $\ln(n\lambda)$ factors are logarithmic relative to the leading $n$ term in the denominator.}
    \label{tab:our_contributions}
\end{table} 
\end{enumerate}

\subsection{Related Work}

Our contributions connect to the broader literature on both DP facility location and intersections of fairness with DP.

\paragraph{DP Facility Location and Median-Finding.} Gupta, Ligett, McSherry, Roth, and Talwar \cite{GLM2009} were among the first to consider the problem of DP facility location across general metric spaces, proposing a solution based on the ``exponential mechanism''~\cite{MT2007} with a social welfare score function. Subsequent work~\cite{JNN2020, FLL2023,EGM2019, LNV2023, MAN2025} has offered tighter runtime and utility bounds, including for relaxed versions of facility location originally proposed by Gupta et~al. Other papers have focused on developing algorithms for related forms of privacy like local DP~\cite{CEF2022} and metric DP~\cite{EMN2023}. Likewise, significant literature has emerged on the closely-related problem of accurate and sample-efficient DP median estimation~\cite{DL2009,GJK2021, TVZ2020, RJC2022, ASSU2023, NRS2011}. However, to our knowledge, we are the first in the domain to examine the \emph{individual-utility loss vector} (which is used in our objective $\fair$) in addition to the more-commonly studied social welfare loss function. Studying these two utility-based metrics simultaneously allows us to characterize a tradeoff not previously addressed.

\paragraph{Fairness under DP.} The existing literature that analyzes fairness under DP constraints has done so with varying notions of fairness. The vast majority of research in this area has been for computer science contexts such as deep learning~\cite{FMS2020, UNK2021}, synthetic data~\cite{GOD2021}, and federated learning~\cite{CZZ2023}, where fairness is typically characterized via standards from the field of algorithmic fairness~\cite{DHP2012}. We also analyze tradeoffs between privacy and fairness, but we do this within a more economic setting, where we use a comparison of agent utilities to measure fairness.

The two works most similar to ours are Pujol~et~al.~\cite{PMK2020} and Tran~et~al.~\cite{TFV2021}, which study the impact of a DP data release (e.g., the release of noisy Census data) on the fairness of downstream economic problems. These papers both examine the tasks of fund allocation and voting rights decision rules. (Pujol~et~al.\ focus on simulations, while Tran~et~al.\ is a theoretical paper.) The measure of ``bias'' or unfairness utilized by both papers is similar in spirit to the one we propose, as it is also based on individual agent utility. We differ from Pujol~et~al.\ and Tran~et~al.\ in two primary ways. First, we analyze the facility location problem, which has a distinct nature compared to allocation and classification, thereby presenting unique challenges in enforcing DP. Second, both previous works focus on mechanisms that follow a particular structure: those that first release summary statistics data in a DP manner before using the noisy data to solve the downstream problem. In this sense, those papers measure how adding DP to an upstream data release impacts the fairness of the later allocation or classification problem (where no additional privacy layers are imposed). Meanwhile, we consider mechanisms that \emph{directly} protect the ``downstream'' facility location problem through DP and assess how much that affects both fairness and social welfare.

\section{Model and Preliminaries}

\subsection{Facility Location Framework}

The canonical one-dimensional facility location problem is defined as follows.

\begin{definition}[One-Dimensional Continuous Facility Location Setting]\label{true_problem_setting}
    Let there be $n$ agents located on a continuous one-dimensional metric space $V=[-m/2, m/2]$ of diameter $m>0$ with metric $d(x,y)=|x-y|$.
    \begin{itemize}
        \item For each agent $i$, let $x_i\in V$ be the agent's location, and denote \emph{a dataset of agent locations} $D\in V^n$ by $D= (x_1, x_2,\ldots, x_n)$. Moreover, assume without loss of generality that the indices are sorted by location, i.e., $x_1\le x_2\le\ldots \le x_n$.
        \item Define utility and social welfare functions $u:V^n\times V\to \mathbb{R}^n$ and $s:V^n\times V\to \mathbb{R}$ that take as input a dataset $D$ and a proposed facility location $\ell$, as follows:
        \begin{itemize}
            \item The utility function outputs an \emph{individual utility vector} $u(D,\ell)\in \mathbb{R}^n$, where the $i$th component $u(D,\ell)_i=-d(x_i,\ell)$ is the utility of agent $i$ given the facility placement $\ell$ and their location $x_i$. 
            \item The social welfare function outputs the overall \emph{social welfare} $s(D, \ell)$ of the placement, measured as the sum of utilities across agents $\sum_{i=1}^n u(D,\ell)_i$.
        \end{itemize}
        \item Assume that the number of agents $n$ is odd. This is solely for ease of exposition; analogous results hold for even $n$ via slight adjustments in theorem statements and proofs.
    \end{itemize}
\end{definition}

While previous work~\cite{GLM2009, JNN2020, EGM2019} on differentially private facility location has examined more general metric spaces, we primarily consider one-dimensional continuous intervals as this constitutes the most fundamental instance of facility location where we can analyze the privacy, welfare, and fairness implications simultaneously.

\begin{definition}[Facility location mechanisms]
    Denote a general \emph{facility location mechanism} by $\M: V^n\to V$, which takes in an input dataset of agent locations $D$ and outputs (possibly in a randomized manner) a facility location $\M(D)\in V$.

    For every dataset $D$ of agent locations, there exists a set $\mathcal{O}_\T(D) =\argmax_{\ell \in V} s(D, \ell)$ of optimal facility placements that maximize the resulting social welfare. This set $\mathcal{O}_\T (D)$ then induces an \emph{optimal facility location mechanism }$\T$, which takes as input a dataset $D$ and outputs a facility location $\T(D)$ selected from $\mathcal{O}_\T(D)$.
\end{definition}

In the one-dimensional model, it is well-known that the set $\mathcal{O_T}(D)$ of optimal facility location(s) occur at or between the most central points:

\begin{proposition}[Optimal facility location is the median]\label{easier_notation}
    Over $V=[-m/2, m/2]$, the optimal set of facility locations for any dataset $D$, sorted as $x_1\le \cdots \le x_n$, is given by \[\mathcal{O_T}(D) =\begin{cases}
        \{x_{\lceil n/2\rceil}\}, &n\text{ odd;}\\
        [x_{\lfloor n/2\rfloor}, x_{\lceil n/2\rceil}], &n\text{ even.}
    \end{cases}\]
\end{proposition}

Proposition~\ref{easier_notation} clarifies why it is helpful notationally to adopt the assumption that $n$ is odd: since $\mathcal{O}_\T(D)$ is a singleton when $n$ is odd, $\T$ is deterministic and uniquely defined.

\subsection{Differential Privacy}
While there are many mechanisms for facility location, we are interested only in ones that satisfy \emph{differential privacy} (DP)~\cite{DMN2006}. Roughly speaking, differential privacy requires that a mechanism treat similar datasets similarly, where we use the \emph{change-one} notion of similarity.
\begin{definition}[Change-one distance]
    Let $D, D'\in V^n$ be two datasets that each contain $n$ agent locations. Viewing $D$ and $D'$ as size-$n$ multisets over elements in $V$, the \emph{change-one distance} $d_{co}: V^n\times V^n\to \mathbb{Z}_{\ge 0}$ between $D$ and $D'$ is given by $d_{co}(D,D') = |D-D'|$ where $-$ denotes multiset difference, i.e., the number of elements (counting multiplicity) that lie in $D$ but not $D'$. Two datasets $D, D'\in V^n$ are \emph{neighboring} if $d_{co}(D,D')=1$.
\end{definition}

The notion of neighboring datasets helps define the concept of differential privacy.

\begin{definition}[Differentially Private Mechanism~\cite{Dwo2006, DMN2006}]\label{dp_def}
    A mechanism $\M: V^n\to \Y$ that takes in datasets from $V^n$ satisfies $\eps$-differential privacy ($\eps$-DP) when \[\Pr [\M(D) \in S] \le e^{\eps} \cdot \Pr[\M(D') \in S]\] for all neighboring datasets $D, D'$ and subsets $S$ of the mechanism's codomain $\Y$.
\end{definition}

We typically require $\eps=O(1)$ so that DP provides meaningful privacy protection. Additionally, we generally assume $n\gg 1/\eps$ for an $\eps$-DP mechanism $\M$; this is necessary for $\M$ to have nontrivial performance on the mechanism's primary goal (in our setting, this would be good social welfare). Note that we use the ``pure" notion of DP like many related papers~\cite{GLM2009, PMK2020, TFV2021}, but similar median-finding problems can behave quite differently under approximate DP~\cite{BNSV2015}.

\subsection{Utility-based metrics for fairness and social welfare}

We are interested in the tradeoff between privacy, social welfare, and fairness within facility location. We require privacy via differential privacy. We now define our metrics for social welfare and fairness. Both metrics compare a DP mechanism $\M$'s output location $\M(D)$ (or more generally, any proposed facility location $\ell$) to the optimal but nonprivate mechanism $\T$'s output location $\T(D)$.

\begin{definition}[$\swdiff$ and $\fair$]\label{loss_metrics} Both of the following functions take as input a dataset of agent locations $D$ and a proposed facility location $\ell$. In both, $\T$ is any optimal mechanism.
    \begin{enumerate}
        \item \emph{Social welfare difference:} \begin{align*}\swdiff (D, \ell) &= s(D, \T(D)) - s(D, \ell).\end{align*} 
        \item \emph{Maximum individual loss in utility:} \[\fair (D, \ell, \T)=\max_{i\in \{1,2,\ldots, n\}} \loss (D, \ell, \T)_i,\] where 
        $\loss (D, \ell, \T)= \E_T[u(D, \T(D))] - u(D, \ell)$ is the individual utility loss vector of switching the facility from $\T(D)$ to $\ell$.
    \end{enumerate}
    
When $\T(D)$ is uniquely defined (i.e., when $\mathcal{O}_T(D)$ is a singleton), we can drop $\T$ as an input for $\fair$ and simply denote $\fair (D,\ell, \T)$ as $\fair (D, \ell)$. In particular, since $\T(D)$ is deterministic by our assumptions that $n$ is odd and $V$ is one-dimensional, we henceforth only write $\fair (D, \ell)$.
\end{definition}

The formulation of $\swdiff$ is standard in the literature for measuring social welfare loss due to DP~\cite{GLM2009}. Meanwhile, $\fair$ is an unstudied metric for unfairness, reminiscent of the loss functions used in~\cite{PMK2020, TFV2021} to analyze DP data releases. $\fair$ accounts for utility losses on an individual level. By penalizing for the worst-off agent's utility loss upon switching from $\T$ to $\ell$, we ensure that no individual is disproportionately harmed by the imposition of privacy.

    Observe that the smallest possible $\swdiff$ and $\fair$ are both $0$, achieved by any optimal facility placement $\T(D)$, and the largest possible $\swdiff$ and $\fair$ are $mn$ and $m$, respectively, where $m$ is the diameter of the space and $n$ is the number of agents.
\section{Overview of results and proofs}\label{summary}

In this section, we summarize our results and offer some intuition behind their proofs. The subsequent sections provide the more formal definitions and technical details.

After deriving closed forms for $\fair$ and $\swdiff$ in Section~\ref{cl_form}, in Section~\ref{impossibility} we demonstrate that a DP mechanism cannot ensure fairness, as measured by $\fair$, for the entire universe of possible datasets of agent locations.

\begin{theorem}[Incompatibility of DP and fairness over all datasets]\label{impossibility_dp_fair_short}
    For every $\eps$-DP facility location mechanism $\M: V^n\to V$ on the interval $V=[-m/2, m/2]$, there exists a dataset $D\in V^n$ of agent locations for which \[\fair (D,\M(D))\ge m/2\] with probability at least $\frac{1}{1+e^\eps}$.
\end{theorem}

\begin{proof}[Proof sketch (see Section~\ref{impossibility} for the full proof)]
    The constructive example for this theorem comes from the pair of datasets $D$ and $D'$ where $D$ has $\lceil n/2\rceil$ agents at $-m/2$ and $\lfloor n/2\rfloor$ agents at $m/2$, and $D'$ has $\lfloor n/2\rfloor$ agents at $-m/2$ and $\lceil n/2\rceil$ agents at $m/2$. As the closed form for $\fair$ in Section~\ref{cl_form} will demonstrate, $\fair$ can be expressed as the distance between the $\M$'s output and the optimal facility location. While $D$ and $D'$ are neighbors, their optimal facility locations differ by $m$. Since $\M$ has similar output distributions over $D$ and $D'$, it must have large $\fair$ on at least one of the two datasets.
    \end{proof}

Since the maximum value of $\fair$ is $m$, a value of $m/2$ is extremely poor. Under the typical expectation that $\eps\le 1$ for DP mechanisms, the above theorem implies that any DP mechanism is bound to have egregious $\fair$ with probability at least $1/(1+\eps)> 25\%$ on some dataset.

Since Theorem~\ref{impossibility_dp_fair_short} says that privacy is inherently incompatible with fairness across general datasets, in Section~\ref{positive_results} we next consider a relaxation of the problem commonly adopted in the DP literature~\cite{DL2009,NRS2011, BCS2018a, BCS2018b, ASSU2023, TVZ2020} when preserving utility over all datasets is infeasible. Specifically, we consider smaller families of datasets and show that we can feasibly obtain good (low) $\fair$ and $\swdiff$ while maintaining differential privacy over \emph{all} datasets. Because the datasets that produce impossibility results are more ``pathological'' than what may we expect of real-world datasets, we might not actually need our DP mechanism to do well over those engineered datasets for it to be generally useful. In particular, we propose two natural families of datasets, $\ctm$ and $\spm$, which guarantee the optimal facility location does not differ drastically across neighboring datasets.

\begin{definition}[$\ctm$ family]
    Call a dataset $D$ \emph{collapsing towards the median }if agents are packed increasingly tightly as they near the median (and optimal facility location) $x_{\lceil n/2\rceil}$:

\begin{itemize}
    \item $|x_{i+1}-x_i|\ge |x_{j+1}-x_j|$ for all $1\le i<j\le \lceil n/2\rceil-1$, and
    \item $|x_i-x_{i-1}|\ge |x_j-x_{j-1}|$ for all $\lceil n/2 \rceil+1\le j<i\le n$.
\end{itemize}

\noindent Define $\ctm$ as the family of all datasets $D\in V^n$ that are collapsing towards the median.
\end{definition}

Datasets that collapse towards the median are appealing because they cleanly ensure concentration around the median agent and the nonexistence of two (or more) peaks of agent density, which help rule out the problematic datasets from the above example. However, the collapsing condition may be too stringent for real-world datasets to fully obey. Thus, we also consider a larger family of datasets that are close to nice absolutely continuous density distributions $P$ over the space $V$. For any distribution $P$, let $F_P$ denote its CDF (so $F_P(\ell) = \Pr_{p\sim P}[p\le \ell]$) and let $f_P$ denote its PDF.

\begin{definition}[$\spm$ family]
    Call a density distribution $P$ over $V$\emph{single-peaked} at $\ell\in V$ if $f_P(x)\le f_P(y)$ for all $x,y\in V$ satisfying either $x\le y\le \ell$ or $x\ge y\ge \ell$. Let $\mathcal{P}$ be the class of all distributions over $V$ that are single-peaked at $F_P^{-1}(0.5)$.
    
    For any fixed $\lambda \ge 0$, define $\spm_\lambda$ as the family of all datasets $D\in V^n$ for which there exists a distribution $P\in \mathcal{P}$ such that the Kolmogorov-Smirnov distance~\cite{Massey1951} between $D$ and $P$ is at most $\lambda$.
\end{definition}

 $\spm$ has a natural distributional interpretation. If one interprets an observed dataset $D$ as a set of $n$ i.i.d.\ samples from a true single-peaked distribution $P\in \mathcal{P}$ of agent locations, then $\spm_\lambda$ for $\lambda = \Theta(1/\sqrt{n})$ is the collection of datasets that, accounting for sampling error, could arise with nontrivial probability.

 While $\ctm$ and $\spm$ are similar to common distributional assumptions in the median-finding literature, they come with their limitations. Notably, the demand for single-peakedness excludes distributions with a stable median but multiple local peaks. Other papers~\cite{NRS2011, TVZ2020} have previously proposed alternative distributional requirements that better handle this case, and a natural follow-up would be to extend our work to those formulations.

 \subsection{Optimal DP Mechanism for $\ctm$ and $\spm$}

In Section~\ref{near_optimal_construct}, we construct a DP facility location mechanism that performs well on both $\swdiff$ and $\fair$ for the two dataset families. We build the mechanism on the (continuous) exponential mechanism of McSherry and Talwar~\cite{MT2007}, which is based on a scoring function with low sensitivity.

\begin{definition}[Global sensitivity with respect to data]\label{def:sensitivity_short}
    For any function $f: V^n \times \mathcal{P}\to \mathbb{R}$ that takes a dataset in $V^n$ and a vector of parameters $p\in \mathcal{P}$, define the \emph{global sensitivity} of $f$ with respect to the data as \[\Delta_f = \max_{p\in \mathcal{P}} \max_{\text{neighboring }D, D'\in V^n} |f(D, p) - f(D', p)|\]
\end{definition}

The sensitivity captures the worst-case pair of neighboring datasets and parameter configuration, where $f$ exhibits the biggest change. The sensitivity of the function $f$ we care about dictates the amount of noise that the DP mechanism $\M$ must have.

\begin{definition}[Continuous Exponential Mechanism]\label{def_cont_exp_mech_short}
    Suppose $s: V^n\times \Y\to \mathbb{R}$ is a scoring function with global sensitivity $\Delta_s$ over datasets $D$ in $V^n$ and outcomes $y$ from a continuous outcome space $\Y$. Define the \emph{continuous exponential mechanism} $\M: V^n\to \Y$ where the PDF $f_{\M, D}$ of $\M(D)$ satisfies \[f_{\M, D}(y) = \frac{\exp(\frac{\eps}{2\Delta_s} s(D, y))}{\int_\mathcal{Y} \exp(\frac{\eps}{2\Delta_s} s(x,z))dz} .\]
\end{definition}

Intuitively, the exponential mechanism seeks to preserve the selected output's performance under the scoring function. To maintain privacy, it does not always take the best-scoring outcome, instead performing a smoothing such that the probability an outcome is selected increases exponentially with its score. In our case, we base the score off of how far the proposed facility placement is from the optimal placement, as is often done in DP mechanisms for the median.

\begin{definition}[Percentile loss function]
    Recall that we denote the ranked ordering of agent locations in $D$ by $x_1\le x_2\le \cdots \le x_n$. Define a \emph{percentile loss function} $q: V^n\times V\to \{0,1,\ldots, \lceil n/2\rceil\}$ given by \[q(D, a) =\begin{cases}
    \min_{i\in \{1,2,\ldots n\}: a\in [x_i, \T(D)]} |\lceil n/2\rceil - i|,&\text{if } a\in [x_1, \T(D)];\\
    \min_{i\in \{1,2,\ldots, n\}: a\in [\T(D), x_i]} |\lceil n/2\rceil - i|,&\text{if } a\in (\T(D), x_n];\\
    \lceil n/2\rceil ,&\text{otherwise}.
\end{cases}\]
\end{definition}

Roughly speaking, $q$ ranks the proposed location $a$ among $x_1,x_2,\ldots, x_n$ and measures the number of points $x_i$ between $a$ and $x_{\lceil n/2\rceil}$. We call $q$ a percentile loss function because $\frac{1}{n}q(D,a)$ measures approximately how many percentile points away $a$ is from $\T(D)$. Among all outputs $\ell \in V$, $\T(D)$ has the \emph{lowest} value with respect to $q$, which makes $q$ a loss function. Therefore, we use the negative of $q$ as the score in an exponential mechanism. We first consider a widened variant of our loss function to account for the continuity of $V$ (see, e.g., \cite{AMS2022}).

\begin{definition}[Widened percentile loss function]\label{def_widened_short}
    For any fixed $\alpha \in [0, 1]$, define the \emph{widened percentile loss function} \[p_\alpha (D, \ell) = \min _{a: |a-\ell|\le \alpha m} q(D,a).\]
\end{definition}

This widening technique is roughly equivalent to binning the continuous space $V$ into bins of width $\alpha m$. It smoothens the loss function over the output range $V$, which is necessary to ensure acceptable performance for datasets $D$ with sharp concentration around $\T(D)$. Consider, for example, a dataset $D^*$ where all $n$ points are stacked at $0$. Note that all outputs $\ell \ne 0$ look indistinguishable under $q$, but the widened $p_\alpha$ would credit a small band $\alpha$ around $0$ as ``good." This allows the following DP mechanism to have sufficient probability mass around the high quality outputs $y$ close to $\T(D)$, even when the set of such $y$ has small measure.

We are now ready to formulate the optimal DP mechanism, which is reminiscent of DP mechanisms for the median~\cite{McS2009}.

\begin{definition}[Widened percentile mechanism]
    For any fixed $\alpha$, define the facility location mechanism $\MoLong$ as an exponential mechanism with score $-p_\alpha$ where the PDF $f_p(D, \ell)$ of $\MoLong (D)$ is given by \[f_p(D, \ell) \propto \exp\left(-\frac \eps 2 p_\alpha(D,\ell)\right).\]
\end{definition} 

\begin{lemma}[$\MoLong $ is $\eps$-DP]\label{mp_is_dp_short}
    For every choice of $\alpha\in [0,1]$, the score $p_\alpha$ has sensitivity 1 and thus the mechanism $\MoLong $ is $\eps$-DP.
\end{lemma}

The proof is in Section~\ref{near_optimal_construct}. As an aside, also note that $\MoLong$ has polynomial runtime $O(n\log n)$, which is not always the case for exponential mechanisms. Sorting the agents is the most expensive operation. Afterwards, scoring and sampling (with the associated probabilities) each only take a linear pass: because $p_\alpha(D,\cdot)$ is piecewise constant with $n+1$ pieces $P_0\cup P_1 \cup \ldots \cup P_n=V$, the output of the exponential mechanism can be generated by first sampling a piece $P_i$ with the correct probability (proportional to $|P_i|\cdot \exp ( -\eps p /2 )$, where $p$ is the value of $p_\alpha(D,\ell)$ for all $\ell\in P_i$) and then sampling uniformly from $P_i$.

\subsection{Performance bounds for $\MoLong$}

In Sections~\ref{upper_bounds} and~\ref{lower_bounds}, we derive high probability upper and lower bounds for $\MoLong$ on our two loss metrics $\fair$ and $\swdiff$, restricting to datasets in $\ctm $ and $\spm_\lambda$. The general framework for doing so is to first bound $p_\alpha$, the loss directly controlling $\MoLong$'s outputs, and then translate this into bounds for $\fair $ and $\swdiff$.
\begin{theorem}[High probability upper bound on $p_\alpha$ for $\MoLong$ over $\spm_\lambda$]\label{p_bound_mp_ddelta_short}
    There exists a universal constant $C$ such that for every $\beta \in (0, 1/3), \eps \in (0,1)$, and $n\in \mathbb{N}$ satisfying $n\eps \ge C\cdot \ln (1/(\alpha\beta))$, it holds for every $D\in \spm_\lambda$ that \[p_{\alpha} (D,\MoLong (D))=O\left(\frac 1\eps \max \left [1, \ln \left(\frac{\lambda}{\alpha \beta n}\right)\right]\right)\] with probability at least $1-\beta$.
\end{theorem}

This theorem provides, under usual DP size assumptions (small $\eps$ and large $n\eps$), a guarantee on the magnitude of $p_\alpha$ with failure probability at most $\beta$.

\begin{proof}[Proof sketch (see Section~\ref{upper_bounds} for the full proof)] The proof stems from (1)~first constructing the dataset with the highest likelihood of producing a large $p_\alpha$ under $\MoLong$, and (2)~demonstrating that even on the worse-case dataset, the value of $p_\alpha$ is bounded. More concretely, for step (1) we can show the following lemma.

\begin{lemma}\label{wc_ddelta_short}
    Fix $\lambda\ge 0$ and let $s=\lceil \lambda n\rceil -1$. For every fixed $k\in \{0,1,\ldots, \lfloor n/2\rfloor\}$, the probability $\Pr[p_{\alpha}(D, \MoLong (D)\le k]$ obtains its minimum across $D\in \spm_\lambda$ on the dataset $D_k^\spm$ with $\lceil n/2\rceil +k$ points at $m/2$, $s$ points at $-m/2$, and the remaining $\lfloor n/2\rfloor -k-s$ points spaced at $-m/2 + i\frac{m}{\lfloor n/2\rfloor -k+\lambda n}$ for $i=\{1,2,\ldots, \lfloor n/2\rfloor -k-s\}$. (See Figure~\ref{fig:wc_dataset_ddelta_short}.)
\end{lemma}

\begin{figure}[h]
    \centering
    \begin{tikzpicture}
        \draw[thick] (-6,0) -- (6,0) node[midway, above] {};
        
        \draw[solid, thick] (-6,-0.2) -- (-6,0.2) node[below=10] {$-m/2$};
        \draw[solid, thick] (-3,-0.1) -- (-3,0.1) node[below=7] {$-m/4$};
        \draw[solid, thick] (0,-0.1) -- (0,0.1) node[below=7] {$\vphantom{/}0$};
        \draw[solid, thick] (3,-0.1) -- (3,0.1) node[below=7] {$m/4$};
        \draw[solid, thick] (6,-0.2) -- (6,0.2) node[below=10] {$m/2$};
        
        \foreach \y in {0.5,0.8,1.1} {
            \filldraw[blue] (6,\y) circle (2pt);
        }
        \node at (6,1.6) {\vdots};
        \foreach \y in {1.9,2.2,2.5} {
            \filldraw[blue] (6,\y) circle (2pt);
        }
        
        \foreach \y in {0.5,0.8} {
            \filldraw[blue] (-6,\y) circle (2pt);
        }
        \node at (-6,1.35) {\vdots};
        \foreach \y in {1.7} {
            \filldraw[blue] (-6,\y) circle (2pt);
        }
        \foreach \x in {-6,-4.8,-3.6,-2.4,-1.2,1.2,2.4} {
            \filldraw[blue] (\x, 0.5) circle (2pt);  
        }

        \node at (0, 0.5) {\textbf{$\cdots$}};

        \draw [thick, decorate,decoration={brace,mirror,amplitude=5pt}] (-6.4,1.8) -- (-6.4,0.4) node[midway,xshift=-12pt] {$s$};
        \draw [thick,decorate,decoration={brace,amplitude=5pt}] (6.4,2.6) -- (6.4,0.4) node[midway,xshift=34pt] {$\lceil n/2 \rceil + k$};
        
        \node[blue] at (6, -1.6) {$\mathcal{T}(D_k^\spm)$};
        \draw[blue, thick,->] (6,-1.3) -- (6,-0.9);
    \end{tikzpicture}
    \caption{Depiction of Dataset $D_k^\spm$, a worst-case dataset in $\spm_\lambda$ for Lemma~\ref{wc_ddelta_short}. Blue points indicate agent locations along $V$.}\label{fig:wc_dataset_ddelta_short}
\end{figure}
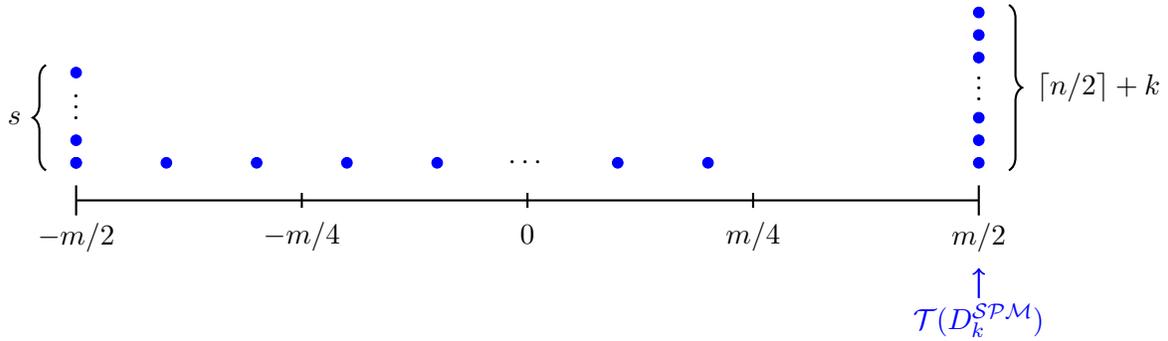

For Step (2), we can then reference the worst-case datasets from Lemma~\ref{wc_ddelta_short}. For any fixed $k$, the probability $p_\alpha(D_k^\spm, \MoLong (D_k^\spm))$ exceeds $k$ is an upper bound on the probability that $p_\alpha(D, \MoLong(D))$ exceeds $k$ for every $D\in \spm_\lambda$. We can then bound the former probability via the structure of the exponential mechanism. Using the geometric formula to collapse similar exponential terms, we can demonstrate that \begin{align*}
        \Pr[p_{\alpha}(D_k^\spm,\MoLong (D_k^\spm))> k] \le \frac{1-\exp(-\frac \eps 2 \lfloor n/2\rfloor)}{1-\exp(-\frac \eps 2)}\cdot \frac{1}{\alpha (\lfloor n/2\rfloor -k)}\exp\left(-\frac \eps 2 (k+1)\right).
    \end{align*}

    Upper bounding the right hand side by a desired failure probability $\beta$ and inverting the inequality demonstrates that for every $D\in V^n$, the mechanism output $\MoLong(D)$ will satisfy \[p_\alpha(D, \MoLong(D))=O\left(\frac 1\eps \max\left[1,\ln \frac{1}{\alpha \beta n \eps}\right]\right)\] with probability at least $1-\beta$.\end{proof}

Optimizing over $\alpha$ and using mathematical relationships between $p_\alpha, \fair$, and $\swdiff$ derived in Section~\ref{cl_form}, we can transfer Theorem~\ref{p_bound_mp_ddelta_short} into the bounds for $\fair$ and $\swdiff$ in Table~\ref{tab:our_contributions}. The full details of these proofs, along with analogous bounds for the $\ctm$ family, are detailed in Section~\ref{upper_bounds}. Note that these bounds are considerable improvements over the one in Theorem~\ref{impossibility_dp_fair_short}. Moreover, these upper bounds are tight, which we show by proving matching lower bounds.

    \begin{theorem}[Lower Bound for $\fair$ over $\spm_\lambda$]\label{final_lower_bound_fair_ddelt_short}
        There exist $\beta>0$ and $C$ such that for every $n\ge 5$, every $\eps\in (0,1)$ satisfying $n\eps \ge C$, and every $\eps$-DP mechanism $\M$, it holds for some $D\in \spm_\lambda$ that \[\fair (D,\M(D))=\Omega \left(\frac{m}{n\eps}\ln \frac 1\beta+m\lambda\right)\] with probability at least $\beta$.
    \end{theorem}

\begin{proof}[Proof sketch (see Section~\ref{lower_bounds} for the full proof)] We construct the lower bounds via the same intuition used in Theorem~\ref{impossibility_dp_fair_short}. We find the most adversarial pairs of datasets $D, D'$ within $\spm_\lambda$ roughly by maximizing the difference between $\T(D)$ and $\T(D')$ while keeping $d_{co}(D, D')$ small. We then use the fact that similar datasets like $D$ and $D'$ must have similar outputs under $\M$, which implies an unavoidably large $\fair$ under one of the datasets. 

In fact, we split the lower bound into two. We first show that $\fair$ is often at least $\Omega (m/(n\eps) \ln (1/\beta))$ on some dataset by considering $D_0, D_\gamma\in \spm_\lambda$ of the visual structure depicted in Figure~\ref{fig:1a}. (For this sketch, we solely provide the visual depiction of these datasets. The rigorous definitions are given in the complete proofs in Section~\ref{lower_bounds}.) We then show that $\fair$ is often at least $\Omega(m\lambda)$ on some dataset by considering neighboring $D_1, D_2\in \spm_\lambda$ with distant optimal facility locations; see Figure~\ref{fig:1b}. Combining the two subresults gives the desired final lower bound.

\begin{figure}[t]
\centering

\begin{subfigure}{.45\linewidth}
\centering
\begin{tikzpicture}[scale=.55]

    \draw[thick] (-6,0) -- (6,0);
    
    \draw[thick] (-6,-0.2) -- (-6,0.2) node[below=4.5] {$\frac{-m}{2}$};
    \draw[thick] (-3,-0.1) -- (-3,0.1) node[below=3] {$\frac{-m}{4}$};
    \draw[thick] (0,-0.1) -- (0,0.1) node[below=3] {$\vphantom{\frac{-0}{1}}0$};
    \draw[thick] (3,-0.1) -- (3,0.1) node[below=3] {$\frac{\vphantom{-}m}{4}$};
    \draw[thick] (6,-0.2) -- (6,0.2) node[below=4.5] {$\frac{\vphantom{-}m}{2}$};
    
    \foreach \x in {-6,-5,-4,-3,-2,-1,1,2,3,4,5,6} {
        \filldraw[blue] (\x, 0.5) circle (2pt);
    }

    \node at (0, 0.5) {\textbf{$\cdots$}};
    
    \node[blue] at (0, -2.5) {$\mathcal{T}(D_0)$};
    \draw[blue, thick,->] (0,-1.8) -- (0,-1.3);
\end{tikzpicture}
\caption{Depiction of Dataset $D_0$. Blue points indicate agent locations along $V$.}
\end{subfigure}
\hfill
\begin{subfigure}{.45\linewidth}
\centering
\begin{tikzpicture}[scale=.55]
    \draw[thick] (-6,0) -- (6,0);
    
    \draw[thick] (-6,-0.2) -- (-6,0.2) node[below=4.5] {$\frac{-m}{2}$};
    \draw[thick] (-3,-0.1) -- (-3,0.1) node[below=3] {$\frac{-m}{4}$};
    \draw[thick] (0,-0.1) -- (0,0.1) node[below=3] {$\vphantom{\frac{-0}{1}}0$};
    \draw[thick] (3,-0.1) -- (3,0.1) node[below=3] {$\frac{\vphantom{-}m}{4}$};
    \draw[thick] (6,-0.2) -- (6,0.2) node[below=4.5] {$\frac{\vphantom{-}m}{2}$};
    \draw[thick] (-2,-0.1) -- (-2,0.1);

    \foreach \y in {0.5,0.8} {
        \filldraw[red] (-2,\y) circle (2pt);
    }
\foreach \y in {1.1, 1.25, 1.4} {
        \node at (-2,\y) {$\cdot$};
    }
    \filldraw[red] (-2,1.7) circle (2pt);

    \foreach \x in {-6,-5,-4,0,1,2} {
        \filldraw[red] (\x, 0.5) circle (2pt);
    }

    \node at (-1, 0.5) {\textbf{$\cdots$}};
    \node at (-3, 0.5) {\textbf{$\cdots$}};
    
    \node[red] at (-1.65, -2.5) {$\mathcal{T}(D_\gamma) = \mathcal{T}(D_0) - \gamma$};
    \draw[red, thick,->] (-2,-1.8) -- (-2,-0.5);

    \draw [thick, decorate, decoration={brace,mirror,amplitude=5pt}] (-2.4,1.8) -- (-2.4,0.4)
        node[midway,xshift=-16pt] {$n_\gamma$};
\end{tikzpicture}
\caption{Depiction of Dataset $D_\gamma$. Red points indicate agent locations along $V$.}
\end{subfigure}

\caption{The datasets $D_0$ and $D_\gamma$ that imply a lower bound of $\Omega(m/(n\varepsilon)\ln(1/\beta))$.}
\label{fig:1a}
\end{figure}
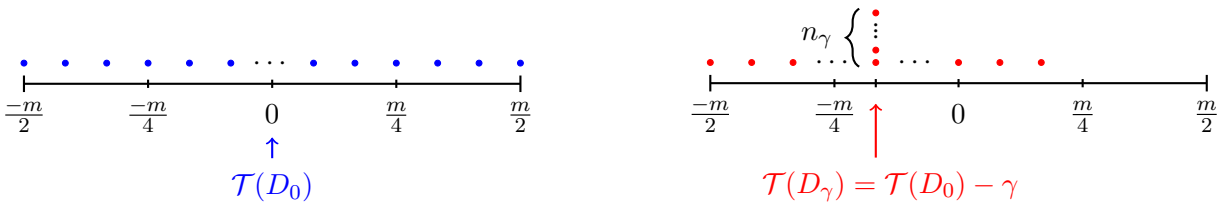

    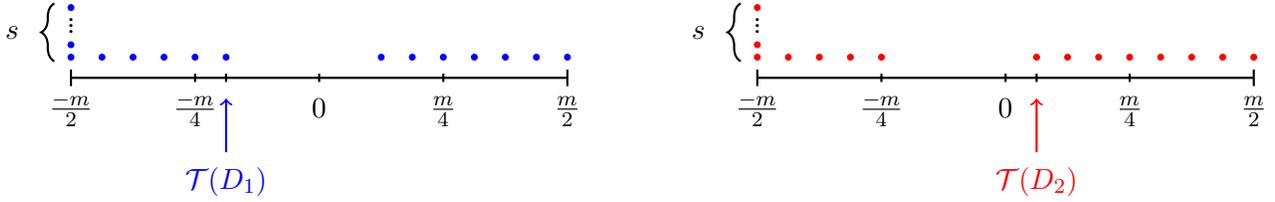
\begin{figure}[t]
\begin{subfigure}{.45\linewidth}
    \centering
\begin{tikzpicture}[scale=.55]
    \draw[thick] (-6,0) -- (6,0) node[midway, above] {};
    
    \draw[solid, thick] (-6,-0.2) -- (-6,0.2) node[below=4.5] {$\frac{-m}{2}$};
    \draw[solid, thick] (-3,-0.1) -- (-3,0.1) node[below=3] {$\frac{-m}{4}$};
    \draw[solid, thick] (0,-0.1) -- (0,0.1) node[below=3] {$\vphantom{\frac{-0}{1}}0$};
    \draw[solid, thick] (3,-0.1) -- (3,0.1) node[below=3] {$\frac{\vphantom{-}m}{4}$};
    \draw[solid, thick] (6,-0.2) -- (6,0.2) node[below=4.5] {$\frac{\vphantom{-}m}{2}$};
    \draw[solid, thick] (-2.25,-0.1) -- (-2.25,0.1) node[below=4mm] {};
    
    \foreach \y in {0.5, 0.8} {
        \filldraw[blue] (-6,\y) circle (2pt);
    }
    \foreach \y in {1.1, 1.25, 1.4} {
        \node at (-6,\y) {$\cdot$};
    }
    \foreach \y in {1.7} {
        \filldraw[blue] (-6,\y) circle (2pt);
    }
    
    \foreach \x in {-5.25,-4.50,-3.75,-3.00, -2.25, 1.50, 2.25, 3.00, 3.75, 4.50, 5.25, 6.00} {
        \filldraw[blue] (\x, 0.5) circle (2pt);  
    }

    \draw [overlay,thick, decorate, decoration={brace,mirror,amplitude=5pt}] (-6.4,1.8) -- (-6.4,0.4) node[midway,xshift=-16pt] {$s$};

    \node[blue] at (-2.25, -2.5) {$\mathcal{T}(D_1) $};
    \draw[blue, thick,->] (-2.25,-1.8) -- (-2.25,-0.5);
\end{tikzpicture}
\caption{Depiction of Dataset $D_1$. Blue points indicate agent locations along $V$.}
\end{subfigure}
\hfill
\begin{subfigure}{.45\linewidth}
\centering
\begin{tikzpicture}[scale=.55]

    \draw[thick] (-6,0) -- (6,0) node[midway, above] {};
    
    \draw[solid, thick] (-6,-0.2) -- (-6,0.2) node[below=4.5] {$\frac{-m}{2}$};
    \draw[solid, thick] (-3,-0.1) -- (-3,0.1) node[below=3] {$\frac{-m}{4}$};
    \draw[solid, thick] (0,-0.1) -- (0,0.1) node[below=3] {$\vphantom{\frac{-0}{1}}0$};
    \draw[solid, thick] (3,-0.1) -- (3,0.1) node[below=3] {$\frac{\vphantom{-}m}{4}$};
    \draw[solid, thick] (6,-0.2) -- (6,0.2) node[below=4.5] {$\frac{\vphantom{-}m}{2}$};
    \draw[solid, thick] (0.75,-0.1) -- (0.75,0.1) node[below=4mm] {};
    
    \foreach \y in {0.5, 0.8} {
        \filldraw[red] (-6,\y) circle (2pt);
    }
    \foreach \y in {1.1, 1.25, 1.4} {
        \node at (-6,\y) {$\cdot$};
    }
    \foreach \y in {1.7} {
        \filldraw[red] (-6,\y) circle (2pt);
    }
    
    \foreach \x in {-5.25,-4.50,-3.75,-3.00, 0.75, 1.50, 2.25, 3.00, 3.75, 4.50, 5.25, 6.00} {
        \filldraw[red] (\x, 0.5) circle (2pt);  
    }

    \draw [overlay,thick, decorate, decoration={brace,mirror,amplitude=5pt}] (-6.4,1.8) -- (-6.4,0.4) node[midway,xshift=-16pt] {$s$};
    
    \node[red] at (0.75, -2.5) {$\mathcal{T}(D_2)$};
    \draw[red, thick,->] (0.75,-1.8) -- (0.75,-0.5);

\end{tikzpicture}
\caption{Depiction of Dataset $D_2$. Red points indicate agent locations along $V$.}
\end{subfigure}
\caption{The datasets $D_1$ and $D_2$ that imply a lower bound of $\Omega(m\lambda)$. Observe that the two datasets share $n-1$ points and only differ in the median agent.}
\label{fig:1b}
    \end{figure}
   \end{proof}

    Similar techniques can be used to prove lower bounds for $\swdiff$ and the $\ctm$ family, and are detailed in Section~\ref{lower_bounds}. As it turns out, the upper and lower bounds on $\fair$ and $\swdiff$ for $\ctm$ match completely, and the bounds for $\spm_\lambda$ match closely up to logarithmic factors under the standard assumptions that $\eps \le 1$ and $\lambda = \Theta(1/\sqrt{n})$. This suggests that $\MoLong$ indeed achieves differential privacy while being near-optimal on fairness and social welfare simultaneously, exhibiting no tradeoff between the two objectives under the imposition of privacy.

    \section{Closed Forms for $\fair$ and $\swdiff$}\label{cl_form}

Before we transition into the technical versions of the proofs previously presented, we first propose closed forms that transform the conceptually intuitive formulations of $\swdiff$ and $\fair$ into mathematically useful ones.

    \begin{theorem}[$\fair$ closed form]\label{fair_closed_form}
    For every dataset $D$ and proposed facility location $\ell \in V=[-m/2, m/2]$, it holds that $\fair (D, \ell) = d(\T(D), \ell)$.
\end{theorem}

Indeed, it turns out that $\fair$ under an outputted facility location $\ell$ can be written cleanly as its distance from the optimal location $\T(D)$.

\begin{proof}
First, suppose $\ell = \T(D)$. Then, \[\fair (D, \M(D), \T(D))=0=d(\T(D), \ell)\] so the claim holds. Now, suppose $\M(D)\ne \T(D)$. Assume without loss of generality that $\M(D)<\T(D)$ as the proof is analogous if the opposite is true. For every $i>\lceil n/2\rceil$, we have that $x_i\ge x_{\lceil n/2\rceil} = \T(D) >\ell$, which means \begin{equation}\label{r_agent_fair}\loss (D,\ell)_i = d(\T(D), \ell).\end{equation}

Meanwhile, for any agent $j<\lceil n/2\rceil$, one may verify that the relations $x_j,\ell\le x_{\lceil n/2\rceil}$ are enough to conclude that \[-d(\T(D), \ell)\le \loss (D,\ell)_j \le d(\T(D), \ell)\] with an equality holding when $x_j=\T(D)$ or $x_j< \ell$. This also immediately implies that \begin{equation}\label{eq:l_agent_fair}|\loss (D,\ell)_j|\le d(\T(D), \ell)=|\loss (D,\ell)_i|.\end{equation} Finally, for the median agent $k$ whose location corresponds to $x_{\lceil n/2\rceil}$, it holds that \begin{equation}\label{eq:med_agent_fair}
\loss(D,\ell)_k = 0-(-d(x_k, \ell)) = d(\T(D), \ell).
\end{equation} Combining Equations~\eqref{r_agent_fair}, \eqref{eq:l_agent_fair}, and \eqref{eq:med_agent_fair} demonstrates that $\fair (D,\ell) = d(\T(D), \ell)$ as claimed.
\end{proof}

Before deriving the closed form for $\swdiff$, we must first characterize the set of agents most negatively impacted by a suboptimal facility placement.

\begin{definition}[Set of crossed agents]\label{crossed_agents}
    For a given dataset $D$ and proposed facility location $\ell\in V=[-m/2, m/2]$, define the set of \emph{crossed agents} $C(D,\ell)$ by \[C(D,\ell) = \begin{cases}
        \{i: i< \lceil n/2\rceil \text{ and } x_i\in [\ell, \T(D)]\},\quad \text{if } \ell <\T(D)\\
        \{i: i> \lceil n/2\rceil\text{ and }x_i\in [\T(D),\ell]\},\quad \text{if } \ell >\T(D)\\
        \emptyset, \quad \text{if } \ell = \T(D)\\
    \end{cases}.\] In words, $C(D,\ell)$ represent the indices of agents that are ``crossed'' when the facility location is moved along $V$ from $\T(D)$ to $\ell$.
\end{definition}

Intuitively, $C(D,\ell)$ is the set of agents we have to ``pay'' for (in units of $\swdiff$) when we move from $\T(D)$ to a suboptimal location $\ell$. To quantify exactly how much we pay, it will be helpful to consider \emph{pairs} of agents located on opposite sides of $\T(D)$ (i.e., one to its left and the other to its right). This is formalized in the proof of the following theorem.

    \begin{theorem}[$\swdiff$ closed form]\label{swdiff_closed_form}
    For every dataset $D$ and proposed facility location $\ell\in V=[-m/2,m/2]$, it holds that \[\swdiff (D,\ell) =d(\T(D),\ell) +2\sum_{j\in C} d(x_j, \ell) \] where $C=C(D,\ell)$ is the set of crossed agents.
\end{theorem}

\begin{proof}
    First, if $\T(D)=\ell$, then $\swdiff (D,\ell) = 0$, $C=\emptyset$, and the relation trivially holds. Now, assume $\T(D)\ne \ell$. Assume without loss of generality that $\ell<\T(D)$ as the alternative case is analogous. Since $\ell < \T(D)$, the set of crossed agents $C(D,\ell)$ is given by \[C(D,\ell) = \{\lceil n/2\rceil -j , \lceil n/2\rceil -(j-1), \ldots, \lceil n/2\rceil -1\}\] for some $j\ge 0$ (if $j=0$, then $C(D,\ell) =\emptyset$). Then, consider pairs of agent locations $x_{\lceil n/2\rceil -i}$ and $x_{\lceil n/2\rceil +i}$. If $\lceil n/2\rceil -i\in C(D,\ell)$ then \begin{align*}
        d(x_{\lceil n/2\rceil-i}, \ell) + d(x_{\lceil n/2\rceil+i}, \ell) &= d(x_{\lceil n/2\rceil-i}, \ell) + [d(x_{\lceil n/2\rceil+i}, \T(D)) +  d(\T(D), x_{\lceil n/2\rceil-i})+ d(x_{\lceil n/2\rceil-i}, \ell)]\\
        &=[d(x_{\lceil n/2\rceil-i}, \T(D)) + d(x_{\lceil n/2\rceil+i}, \T(D))] + 2d(x_{\lceil n/2\rceil-i}, \ell).
    \end{align*} Otherwise, if $\lceil n/2\rceil-i\not\in C(D,\ell),$ then \begin{align*}
        d(x_{\lceil n/2\rceil-i}, \ell) + d(x_{\lceil n/2\rceil+i}, \ell) &= [d(x_{\lceil n/2\rceil-i}, \T(D)) - d(\T(D), \ell)] + [d(x_{\lceil n/2\rceil+i}, \T(D)) + d(\T(D), \ell)]\\
        &= d(x_{\lceil n/2\rceil-i}, \T(D)) + d(x_{\lceil n/2\rceil+i}, \T(D)).
    \end{align*} The intuition for the above computations is that if one of the agents in a pairing belongs to the crossed set, then some additional utility loss is incurred under the move from $\T(D)$ to $\ell$. However, if neither agent in the pairing is crossed, then their total \emph{summed} utility stays invariant under the move from $\T(D)$ to $\ell$.
    Therefore, \begin{align*}
        s(D, \ell) &= -d(x_{\lceil n/2\rceil}, \ell) - \sum_{i=1}^{\lfloor n/2\rfloor} [d(x_{\lceil n/2\rceil-i}, \ell) + d(x_{\lceil n/2\rceil+i},\ell)]\\
        &= -d(x_{\lceil n/2\rceil},\ell) - \sum_{i=1}^{\lfloor n/2\rfloor} [d(x_{\lceil n/2\rceil-i}, \T(D)) + d(x_{\lceil n/2\rceil+i},\T(D))]-\sum_{j\in C} 2d(x_j, \ell)
    \end{align*} and $s(D,\T(D)) = -\sum_{i=1}^{\lfloor n/2\rfloor} [d(x_{\lceil n/2\rceil-i}, \T(D)) + d(x_{\lceil n/2\rceil+i},\T(D))]$ so \[\swdiff (D,\ell) =s(D,\T(D)) -s(D,\ell) = d(x_{\lceil n/2\rceil},\ell)+\sum_{j\in C} 2d(x_j, \ell)\] as required.
\end{proof}

\section{General incompatibility of privacy and fairness}\label{impossibility}

We now turn towards concretely analyzing how DP mechanisms perform with respect to the metrics of $\swdiff$ and $\fair$. We first demonstrate that a DP mechanism cannot possibly do well on fairness for all possible input datasets.

\begin{theorem}[Incompatibility of DP and fairness over all datasets]\label{impossibility_dp_fair}
    For every $\eps$-DP facility location mechanism $\M: V^n\to V$ on the interval $V=[-m/2, m/2]$, there exists a dataset $D$ of agent locations for which \[\fair (D,\M(D))\ge m/2\] with probability at least $\frac{1}{1+e^\eps}$.
\end{theorem}

\begin{proof}
    We consider a pair of datasets $D$ and $D'$ that are commonly used to prove accuracy lower bounds for DP medians. Let $D$ be a dataset with $\lceil n/2\rceil$ agents at $-m/2$ and $\lfloor n/2\rfloor$ agents at $m/2$, and let $D'$ be a dataset with $\lfloor n/2\rfloor$ agents at $-m/2$ and $\lceil n/2\rceil$ agents at $m/2$. Note that $D$ and $D'$ can be viewed as neighbors since we can transform $D$ into $D'$ by moving one agent from $-m/2$ to $m/2$.
    
    Since $\T(D)=-m/2$ and $\T(D')=m/2$, the closed form for $\fair$ from Theorem~\ref{fair_closed_form} implies that $\fair (D,\M(D))$ is at least $m/2$ when $\M(D)\ge 0$ and $\fair(D',\M(D'))$ is at least $m/2$ when $\M(D')\le 0$. However, it holds by DP that \[e^{-\eps}\Pr[\M(D)\le 0]\le \Pr[\M(D') \le 0]\] so \begin{align*}
        \min (\Pr[\M(D)\ge 0], \Pr[\M(D')\le 0]) &\ge \min ( 1-\Pr[\M(D)\le 0], e^{-\eps}\Pr[\M(D)\le 0])\ge \frac{e^{-\eps}}{1+e^{-\eps}}
    \end{align*} where the second inequality follows by setting $\Pr[\M(D)\le 0]$ equal to $\frac{1}{1+e^{-\eps}}$ to make the two terms in the min expression coincide.
\end{proof}

\section{Positive results on restricted datasets}\label{positive_results}

Theorem~\ref{impossibility_dp_fair} shows that adding privacy inevitably produces additional unfairness, meaning that privacy and fairness are generally incompatible.  This motivates us to instead consider a relaxation that is common in the DP literature~\cite{DL2009,NRS2011, BCS2018a, BCS2018b, ASSU2023, TVZ2020} when maintaining both privacy and good utility is infeasible for all datasets. In particular, we still require differential privacy over all datasets, but now instead require good $\fair$ and $\swdiff$ over smaller families of datasets. Such a relaxation is reasonable because the datasets that produce impossibility results are more ``pathological" than real-world datasets, so it is not practically necessary to require good welfare performance over those instances.

Note that while we only seek good utility over a subset of datasets, we still demand privacy over all datasets. The privacy constraint is stronger because DP fundamentally serves to protect outliers who do not fit any distributional assumptions. A mechanism that satisfying this pair of objectives unequivocally protects privacy, while also granting good quality as long as the input dataset is relatively well-behaved.

Within our context, the poor performance on $\fair$ is unavoidable for datasets where the median jumps significantly upon moving to its neighbors. Thus, we primarily consider utility of DP mechanisms over families of datasets where the median $\T(D)$ is robust. We suggest two such families: one that strictly enforces the properties we desire, and another that approximately maintains the same properties but allows for more freedom -- making it more reminiscent of real-world data.

While we already introduced these families in Section~\ref{summary}, we offer more detailed definitions and intuition here.

\subsection{$\ctm$: Family of collapsing datasets}

For a dataset $D$, recall that we order its $n$ agent locations as $x_1\le x_2\le \cdots \le x_n$. Our first family of datasets asks agents to be packed increasingly tightly as they near the median location $x_{\lceil n/2\rceil}$.

\begin{definition}[$\ctm$ family]
    Call a dataset $D$ \emph{collapsing towards the median }if the following two conditions hold:

\begin{itemize}
    \item $|x_{i+1}-x_i|\ge |x_{j+1}-x_j|$ for all $1\le i<j\le \lceil n/2\rceil-1$
    \item $|x_i-x_{i-1}|\ge |x_j-x_{j-1}|$ for all $\lceil n/2 \rceil+1\le j<i\le n$
\end{itemize}

\noindent Define $\ctm$ as the family of all datasets $D\in V^n$ that are collapsing towards the median.
\end{definition}

Intuitively, collapsing datasets are appealing because they ensure the density of agents increases as we move towards the median on either side. For any collapsing dataset $D$, a histogram (Figure~\ref{fig:dc_hist}) of the number of agents within each segment of $V$ produces a shape that is roughly single-peaked at $\T(D)$. This assures concentration around the median agent and the nonexistence of two (or more) ``peaks'' of agent density, which help rule out the problematic datasets from Section~\ref{impossibility}.

\begin{figure}[h]
    \centering
    \includegraphics[width=0.8\textwidth]{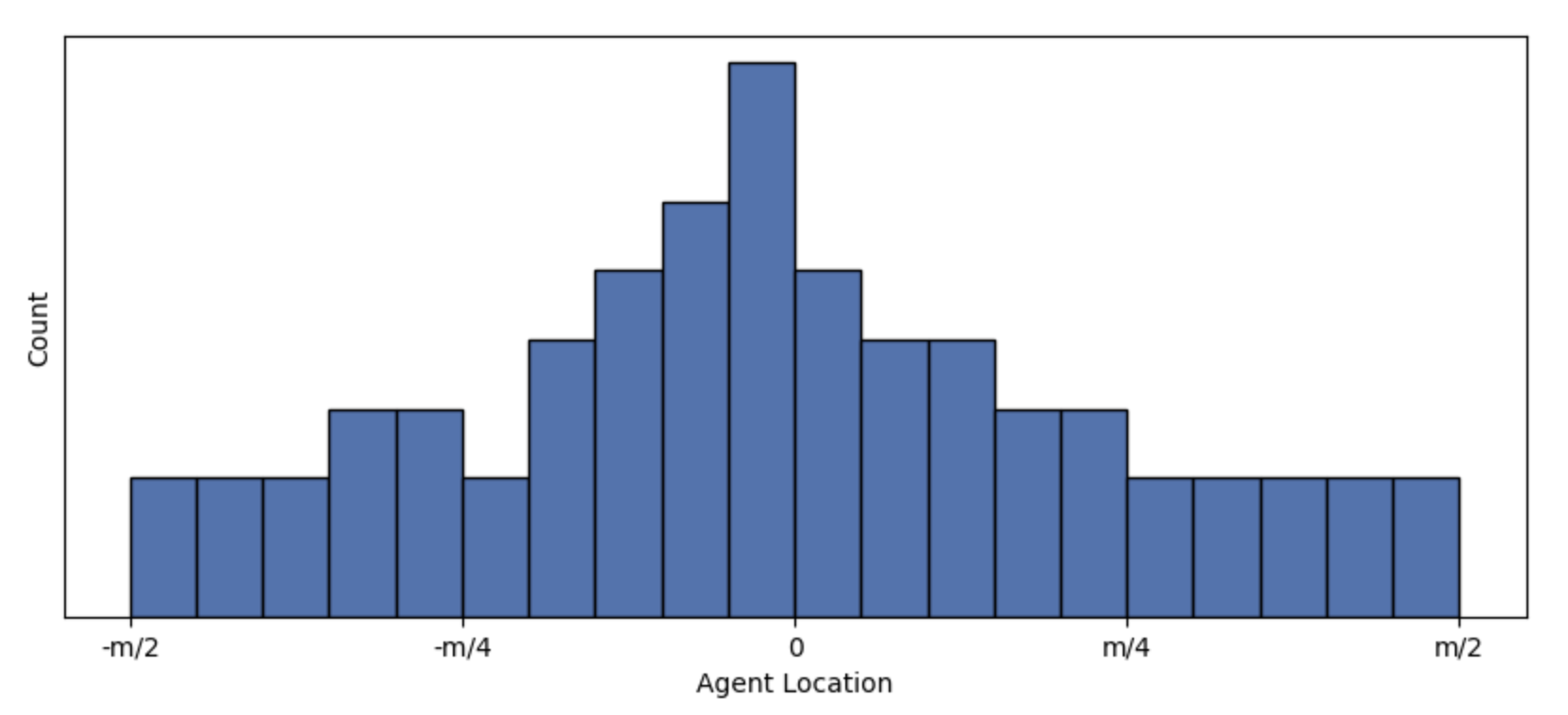}
    \caption{Histogram of a Collapsing Dataset}
    \label{fig:dc_hist}
\end{figure}

The family of collapsing datasets $\ctm$ directly establishes high concentration about the median, ensuring low $\fair$ and $\swdiff$ can reasonably be expected. However, its collapsing condition may be too stringent for real-world datasets to fully obey. Therefore, we can consider a looser requirement that still preserves the fundamental structural properties from $\ctm$ but is easier to satisfy.

\subsection{$\spm_\lambda$: Family of datasets close to single-peakedness at the median}

In our second family of datasets, we ask datasets to approximately increase in density towards the median, but allow for some fluctuations that are inevitably present in naturally-occurring datasets. We can model this by looking for datasets that are ``close'' to a desirable \emph{continuous} density distribution over $V$.

More concretely, consider absolutely continuous probability distributions $P$ over the space $V$, where the density of $P$ at any $\ell\in V$ represents the relative likelihood that a random sample drawn from $P$ would equal $\ell$. For any distribution $P$, let $F_P$ denote its CDF (so $F_P(\ell) = \Pr_{p\sim P}[p\le \ell]$) and let $f_P$ denote its PDF.

\begin{definition}[Density distributions single-peaked at the median]
    Call a density distribution $P$ over $V$\emph{single-peaked} at $\ell\in V$ if $f_P(x)\le f_P(y)$ for all $x,y\in V$ satisfying either $x\le y\le \ell$ or $x\ge y\ge \ell$. Let $\mathcal{P}$ be the class of all distributions over $V$ that are single-peaked at $F_P^{-1}(0.5)$.
\end{definition}

Note that $\mathcal{P}$ mimics the desirable properties of $\ctm$ by ensuring its distributions have density that increases monotonically towards the median. If the actual dataset of agent locations looks similar to the continuous distribution from $\mathcal{P}$, then it will also approximately inherit this property. To quantify what ``looks similar'' means here, we utilize the Kolmogorov-Smirnov measure of distance between distributions.

\begin{definition}[Kolmogorov-Smirnov Distance~\cite{Massey1951}]
    Let $D, D'$ be two distributions over the reals, and let $F, F': \mathbb{R}\to [0,1]$ be their respective CDFs. The \emph{Kolmogorov-Smirnov (K-S) distance} between $F$ and $F'$ is given by \[\mathrm{KS}(F, F') = \sup_{x\in \mathbb{R}} |F(x) - F'(x)|.\] For ease of notation, we will often also refer to the K-S distance between $D$ and $D'$ as $\mathrm{KS}(F, F')$.
\end{definition}

Note that any dataset $D\in V^n$ can be treated as a distribution over $V$ as well. In particular, its CDF $F$ at any $\ell \in V$ is given by $F(\ell) =\frac{1}{n} |\{i: x_i \le \ell\}|$. Therefore, using K-S as the distance between a dataset and true distribution, we can characterize a new family of datasets.

\begin{definition}[$\spm$ family]
    For any fixed $\lambda \ge 0$, define $\spm_\lambda$ as the family of all datasets $D\in V^n$ for which there exists a distribution $P\in \mathcal{P}$ such that $\mathrm{KS}(D, P)\le \lambda$.
\end{definition}

This family captures all datasets that are not too far from a continuous single-peaked distribution in CDF at any point. Since every $P\in \mathcal{P}$ has the necessary high concentration near its median, each $D\in \spm_\lambda$ should also have comparable properties up to errors proportional to $\lambda$.

\paragraph{Data-Generative Interpretation of $\spm_\lambda$.} Another reason $\spm_\lambda$ is that it has a natural distributional interpretation. If one views an observed dataset as a collection of $n$ i.i.d. draws from a true single-peaked distribution $P\in \mathcal{P}$ of agent locations, then $\spm_\lambda$ (for suitable $\lambda$) encapsulates all datasets that could be generated with nontrivial probability.

\begin{theorem}[Dvoretzky-Kiefer-Wolfowitz Inequality]\label{dkw}
    Let $D$ be a probability distribution with CDF $F$. Let $D_n$ comprise of $n$ independent draws from $D$ and let $F_n$ be the empirical CDF of $D_n$. Then, \[\Pr [\mathrm{KS}(F, F_n)>\lambda] \le 2e^{-2n\lambda^2}\] for any positive $\eps$. Specifically, for $\lambda = t\sqrt{2/n}$ with any constant $t$, we have \[\Pr\left[\mathrm{KS}(F, F_n) > \frac{t\sqrt{2}}{\sqrt{n}}\right] \le 2e^{-t}\]
\end{theorem}

The DKW inequality provides a bound on the K-S distance between a true distribution and an empirical distribution generated from it. In particular, the K-S distance will be $O(1/\sqrt{n})$ with high probability. This implies the following corollary.

\begin{corollary}[Data-Generating Distribution Interpretation of $\spm_\lambda$]\label{ddelta_datagen}
    For every constant $\beta > 0$, there exists $\lambda ^*= O(1/\sqrt{n})$ such that for every ``true" probability distribution $P\in \mathcal{P}$, any dataset $D$ of $n$ agent locations sampled from $\mathcal{P}$ will lie in $\spm_{\lambda^*}$ with probability at least $1-\beta$.
\end{corollary}

It is also worth noting that $\spm_\lambda$ is a more general family compared to $\ctm$, assuming a realistically large $\lambda$.

\begin{theorem}\label{dc_is_in_ddelta}
    For $\lambda\ge 1/(n-1)$, the family of collapsing datasets $\ctm$ is contained in $\spm_\lambda$. That is, $\ctm\subseteq \spm_\lambda$.
\end{theorem}

\begin{proof}
    For any collapsing dataset $D\in \ctm$, consider the following distribution $P\in \mathcal{P}$ with PDF $f$. For $\ell \in [x_i, x_{i+1})$ for any $i\in \{1,\ldots, n-1\}$, define \[f(\ell) = \frac{1}{(n-1)|x_{i+1}-x_i|}\] and for any $\ell<x_1$ or $\ell > x_n$, set $f(\ell) = 0$. Let $F, F_n$ be the CDFs of $P$ and $D$ respectively. Note that by the collapsing property, $D$ is single-peaked at $\T(D)$. Now, for any $\ell$, \[|F(\ell) -F_n(\ell)|\le \max_{k\in \{0,\ldots, n-1\}}\max\left(\left|\frac{k+1}{n}-\frac{k}{n-1}\right|, \left|\frac{k}{n-1}-\frac{k}{n}\right|\right)\le \frac{1}{n-1},\] so the K-S distance between $P$ and $D$ is at most $1/(n-1)$. Thus, $P$ will lie in $\spm_\lambda$ for $\lambda \ge 1/(n-1)$.
\end{proof}

Note that the finiteness of any $D\in \spm_\lambda$ means that jumps in multiples of $1/n$ are unavoidable in its CDF, so any nontrivial $\spm_\lambda$ will require $\lambda = \Omega (1/n)$. Moreover, per the sampling interpretation of $\spm_\lambda$ described by Corollary~\ref{ddelta_datagen}, we typically will consider $\lambda = \Theta(1/\sqrt{n})$. Therefore, we will assume henceforth that $\lambda=\Omega(1/n)$ and so Theorem~\ref{dc_is_in_ddelta} effectively demonstrates that $\ctm$ is a subset of $D_{\lambda}$.

\subsection{A Brief Comparison of $\ctm$ and $\spm_\lambda$}

Together, $\ctm$ and $\spm_\lambda$ represent two families of datasets with properties that enable good performance by DP mechanisms. In upcoming sections, we will thus consider upper and lower bounds on $\fair$ and $\swdiff$ over these families. The analysis of each family provides distinct insights into the broader question of trade-offs between privacy, social welfare, and fairness.

$\ctm$ is the smaller and more ``synthetic'' family that enforces stricter requirements on its datasets. Relative to those in $\spm_\lambda$, the datasets in $\ctm$ will behave better under DP mechanisms and exhibit lower $\fair$ and $\swdiff$. Analyzing $\ctm$ most cleanly reveals the fundamental intuitions of how DP mechanisms impact facility location datasets.

Meanwhile, $\spm_\lambda$ is the larger family and a more compelling representation of real-world data, particularly through its interpretation as the family of datasets generated as samples from a true distribution. The (lack of) tradeoff we find between $\fair$ and $\swdiff$ we find over $\spm_\lambda$ allows us to draw a more general conclusion to our question.

\section{Near-optimal mechanism over the relaxation}

We now provide the concrete results under our relaxed families $\ctm$ and $\spm_\lambda$. We propose a near-optimal DP mechanism $\MoLong$ and quantify upper and lower bounds on its performance with respect to $\fair$ and $\swdiff$.

\subsection{Mechanism construction}\label{near_optimal_construct}

We build our mechanism off the (continuous) exponential mechanism from McSherry and Talwar~\cite{MT2007}. We begin by defining the concept of global sensitivity that is fundamental to the construction of most DP mechanisms.

\begin{definition}[Global sensitivity with respect to data]\label{def:sensitivity}
    For any function $f: V^n \times \mathcal{P}\to \mathbb{R}$ that takes in a dataset from $V^n$ and a vector of parameters $p\in \mathcal{P}$, define the \emph{global sensitivity} of $f$ with respect to the data as \[\Delta_f = \max_{p\in \mathcal{P}} \max_{\text{neighboring }D, D'\in V^n} |f(D, p) - f(D', p)|\]
\end{definition}

The sensitivity measures how much $f$ can change across a pair of neighboring datasets under the ``worst-case" set of parameters. If a DP mechanism is attempting to ``perform well" with respect to a function $f$, it will need to add noise proportional to the sensitivity of $f$ to preserve privacy.

\begin{definition}[Continuous Exponential Mechanism]\label{def_cont_exp_mech}
    Suppose $s: V^n\times \Y\to \mathbb{R}$ is a scoring function with global sensitivity $\Delta_s$ over datasets $D$ from $V^n$ and ``outcomes'' $y$ from a continuous outcome space $\Y$. Define the \emph{continuous exponential mechanism} $\M: V^n\to \Y$ where the PDF $f_{\M, D}$ of $\M(D)$ satisfies \[f_{\M, D}(y) = \frac{\exp(\frac{\eps}{2\Delta_s} s(D, y))}{\int_\mathcal{Y} \exp(\frac{\eps}{2\Delta_s} s(x,z))dz} .\]
\end{definition}

Intuitively, the exponential mechanism seeks to preserve the selected output's performance under the scoring function. To maintain privacy, it will not always take the best-scoring outcome, instead performing a smoothing such that the probability an outcome is selected increases exponentially with its score. In our case, we will base the score off of how ``far" the proposed facility placement is from the optimal placement, as is often done in DP mechanisms for the median.

\begin{definition}[Percentile loss function]
    Recall that we denote the ranked ordering of agent locations in $D$ by $x_1\le x_2\le \cdots \le x_n$. Define a \emph{``percentile'' loss function} $q: V^n\times V\to \{0,1,\ldots, \lceil n/2\rceil\}$ given by \[q(D, a) =\begin{cases}
    \min_{i\in \{1,2,\ldots n\}: a\in [x_i, \T(D)]} |\lceil n/2\rceil - i|,\quad \text{if } a\in [x_1, \T(D)]\\
    \min_{i\in \{1,2,\ldots, n\}: a\in [\T(D), x_i]} |\lceil n/2\rceil - i|,\quad \text{if } a\in (\T(D), x_n]\\
    \lceil n/2\rceil ,\quad \text{otherwise}
\end{cases}\]
\end{definition}

$q$ roughly represents the number of agents $x_i$ between a proposed location $a$ and the optimal location $x_{\lceil n/2\rceil}$, hence inducing the label of being a ``percentile" loss function. Among all output choices $\ell \in V$, $\T(D)$ has the \emph{lowest} value with respect to $q$, which makes $q$ a loss function. Therefore, to apply the exponential mechanism, we will ultimately consider the negative of $q$ as the score to maximize. 

The reason for basing our score function off $q$ is two-fold. Firstly, $q$ has low sensitivity, so the relative amount of noise the mechanism adds is low. Secondly, $q$ accounts for the positioning of all the agents in its scoring, so it serves as a good proxy for not only $\fair$ but also $\swdiff$ (which is more sensitive to the locations of non-median agents).

However, to finalize our score function, we need to add a ``widening'' feature to account for the continuity of $V$, as previously done in other papers such as \cite{AMS2022}.

\begin{definition}[Widened percentile loss function]\label{def_widened}
    For any fixed $\alpha \in [0, 1]$, define the ``widened'' percentile loss function \[p_\alpha (D, \ell) = \min _{a: |a-\ell|\le \alpha m} q(D,a).\]
\end{definition}

This widening technique ensures that datasets $D$ with sharp concentration around $\T(D)$ will not experience dramatically higher loss as soon as the possible output $\ell$ moves away from $\T(D)$. To explicitly compare $p_\alpha$ and $q$ as loss functions for a facility location exponential mechanism $\M$, consider $D^*$ with all $n$ points stacked at $0$ and note that $q(D^*,\ell) = \lceil n/2\rceil$ for all $\ell\ne 0$. Moreover, the continuous nature of the exponential mechanism $\M$ would imply that $\Pr[ \M(D^*) = 0]=0$. Together, these two characteristics mean $\M$ with score $-q$ selects its output uniformly from $V$, which is clearly undesirable. Meanwhile, $p_\alpha$ ensures that the entire interval $I$ of length $2\alpha m$ around $\T(D^*)$ will attain the minimum possible loss, so $\M$ with score $-p_\alpha$ will ensure $\M(D^*)$ falls in the ``good'' interval $I$ with disproportionate probability.

Finally, we propose our primary DP mechanism $\MoLong$, which follows a similar structure to DP mechanisms for the median.~\cite{McS2009}

\begin{definition}[Widened percentile mechanism $\MoLong$]
    For any fixed $\alpha$, define the facility location mechanism $\MoLong$ as an exponential mechanism with score $-p_\alpha$ where the PDF $f_p(D, \ell)$ of $\MoLong (D)$ is given by \[f_p(D, \ell) \propto \exp\left(-\frac \eps 2 p_\alpha(D,\ell)\right).\]
\end{definition} 

\begin{lemma}[$\MoLong $ is $\eps$-DP]\label{mp_is_dp}
    For every choice of $\alpha\in [0,1]$, the mechanism $\MoLong $ is $\eps$-DP.
\end{lemma}

\begin{proof}
    Since $\MoLong$ is an exponential mechanism, it simply suffices to verify that $p_\alpha$ has sensitivity 1 to prove $\MoLong $ is $\eps$-DP. Fix any $a\in V$ and any dataset $D$, and assume without loss of generality that $a\le x_{\lceil n/2\rceil}$. Let $S$ be the collection of indices $i$ for which $x_i$ ``contributes'' to the score $q(D,a)$, so either $S=\emptyset$ (if $a=x_{\lceil n/2\rceil}$) or $S=\{s,s+1,\ldots, \lceil n/2\rceil\}$ for some $s\le \lceil n/2\rceil$. Note that $q(D,a) = |S|$ and $S$ can only increase or decrease by one member upon switching to a neighboring $D'$, which means $q(D,a)$ changes by at most 1. Therefore, \[p_\alpha (D', \ell) =\min_{a:|a-\ell|\le \alpha m} q(D', a) \le \min_{a:|a-\ell|\le \alpha m} (q(D,a)+1)=p_\alpha(D,\ell)+1\] and similarly \[p_{\alpha} (D',\ell)\ge p_{\alpha} (D,\ell)-1\] so the sensitivity of $p_\alpha$ is 1 as claimed. Hence, $\MoLong$ is $\eps$-DP.
\end{proof}

\subsection{High probability upper bounds on loss}\label{upper_bounds}

We demonstrate that $\MoLong$ is a good DP mechanism with respect to both $\swdiff$ and $\fair$ by deriving high probability upper bounds on the performance of $\MoLong$. We will generally prove our upper bounds via the following structure.

\begin{procedure}\label{procedure}[Deriving high probability upper bounds for $\fair$ and $\swdiff$]$\quad$
    \begin{enumerate}
    \item Derive a high probability upper bound on the value of $p_\alpha(D,\MoLong (D))$. \label{procedure_step_1}
    \item Build a relationship between $p_{\alpha} $ and $\fair$ to construct a high probability upper bound on $\fair (D, \MoLong(D))$ from the result of Step 1.\label{procedure_step_2}
    \item Build relationships from $p_{\alpha}$ and $\fair$ to $\swdiff$ to construct a high probability bound on $\swdiff (D,\MoLong(D))$ from the results of Steps 1 and 2.\label{procedure_step_3}
\end{enumerate}
\end{procedure}

We start with upper bounds over $\ctm$ and then follow a similar procedure over $\spm_\lambda$.

\subsubsection{Upper bounds over $\ctm$}\label{perf_dc}

First, we address Step~\ref{procedure_step_1} of Procedure~\ref{procedure} by finding an upper bound on $p_\alpha$ for $\MoLong$. We utilize a technical lemma which characterizes the ``worst-case'' datasets in $\ctm$, under which it is most likely that $p_{\alpha}(D,\MoLong (D))$ is large.

\begin{lemma}\label{wc_dc}
    For every fixed $k\in \{0, 1,\ldots, \lfloor n/2\rfloor\}$, the probability $\Pr[p_\alpha(D, \MoLong (D))\le k]$ obtains its minimum across $D\in \ctm$ on the dataset $D_k^\ctm$ with $\lceil n/2\rceil + k$ points at $m/2$ and the remaining $\lfloor n/2\rfloor -k$ points uniformly spaced at $m/2 - i\frac{m}{\lfloor n/2\rfloor -k}$ for $i\in \{1, 2,\ldots, \lfloor n/2\rfloor - k\}$.
\end{lemma}

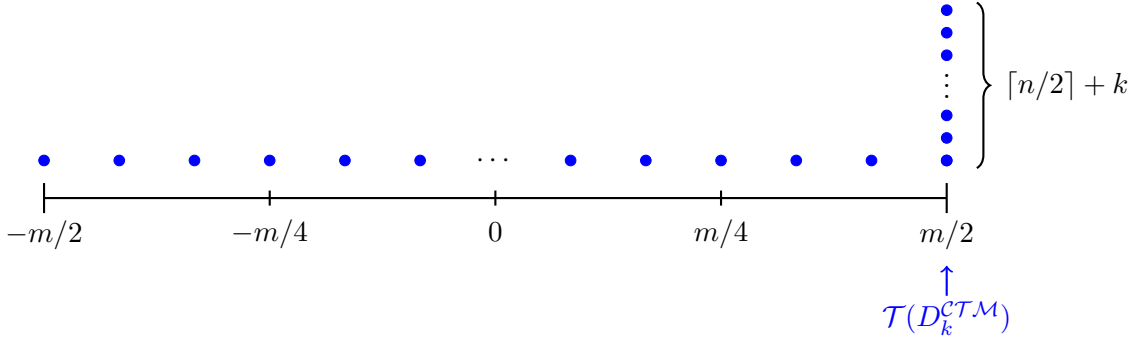
\begin{figure}[h]
    \centering
    \begin{tikzpicture}

        \draw[thick] (-6,0) -- (6,0) node[midway, above] {};
        
        \draw[solid, thick] (-6,-0.2) -- (-6,0.2) node[below=10] {$-m/2$};
        \draw[solid, thick] (-3,-0.1) -- (-3,0.1) node[below=7] {$-m/4$};
        \draw[solid, thick] (0,-0.1) -- (0,0.1) node[below=7] {$\vphantom{/}0$};
        \draw[solid, thick] (3,-0.1) -- (3,0.1) node[below=7] {$m/4$};
        \draw[solid, thick] (6,-0.2) -- (6,0.2) node[below=10] {$m/2$};
        
        \foreach \y in {0.5,0.8,1.1} {
            \filldraw[blue] (6,\y) circle (2pt);
        }
        \node at (6,1.6) {\vdots};
        \foreach \y in {1.9,2.2, 2.5} {
            \filldraw[blue] (6,\y) circle (2pt);
        }
        
        \foreach \x in {-6,-5,-4,-3,-2,-1,1,2,3,4,5,6} {
            \filldraw[blue] (\x, 0.5) circle (2pt);  
        }

        \node at (0, 0.5) {\textbf{$\cdots$}};

        \draw [thick,decorate,decoration={brace,amplitude=5pt}] (6.4,2.6) -- (6.4,0.4) node[midway,xshift=34pt] {$\lceil n/2 \rceil + k$};
        
        \node[blue] at (6, -1.6) {$\mathcal{T}(D_k^\ctm)$};
        \draw[blue, thick,->] (6,-1.3) -- (6,-0.9);
    \end{tikzpicture}
    \caption{Depiction of Dataset $D_k^\ctm$, a worst-case dataset in $\ctm$ for Lemma~\ref{wc_dc}. Blue points indicate agent locations along $V$.}
    \label{fig:wc_dataset_dc}
\end{figure}

\begin{proof}
    For brevity, denote the $\MoLong$ mechanism by $\Mo$ throughout this proof. We first find a lower bound on $\Pr[p_{\alpha} (D,\Mo (D))\le k]$ and then show that $D_k^\ctm$ achieves it. For any $D\in \ctm$, let $L_D, R_D$ denote the events $\Mo (D)<\T(D)$ and $\Mo (D)>\T(D)$ respectively. Since $\Pr[p_\alpha (D, \Mo (D)) = \T(D)]=0$, it holds that \begin{align*}
    \Pr[p_\alpha (D, \Mo (D))\le k] &= \Pr[L_D] 
    \cdot \Pr\big[p_\alpha (D, \Mo (D))\le k \,\big|\, L_D\big] \\
    &\quad + \Pr[R_D] 
    \cdot \Pr\big[p_\alpha (D, \Mo (D))\le k \,\big|\, R_D\big] \\
    &\ge \min \{\Pr\big[p_\alpha (D, \Mo (D))\le k \,\big|\, L_D\big], \Pr\big[p_\alpha (D, \Mo (D))\le k \,\big|\, R_D\big]\}
\end{align*}

Therefore, it suffices to bound the probability of $p_{\alpha} (D,\Mo (D))\le k$ occurring given $\Mo (D)$ after fixing the side of $\T(D)$ that $\Mo (D)$ is on. Assume without loss of generality that we fix $\Mo (D)$ to be on the left of $\T(D)$, and assume henceforth that all events are already conditioned on $\Mo (D) <\T(D)$.

Observe that for any $D\in \ctm$, the interval $(\T(D)-\alpha m, \T(D))$ of length $\alpha m$ must satisfy $p_\alpha (D, \ell) = 0$ for any $\ell\in (\T(D)-\alpha m, \T(D))$. Hence, $p_\alpha (D, \Mo (D))\le k$ always has unnormalized mass under $\Mo $'s weighting of at least $\alpha m \exp(-\eps / 2 \cdot 0) = \alpha$. Equality holds when $x_{\lfloor n/2\rfloor - (k-1)}=\T(D)$.

    Additionally, note that $I=[-m/2, x_{\lfloor n/2\rfloor -(k-1)} -\alpha m]$ is the set of points $\ell$ for which $p_{\alpha}(D,\ell)>k$. Fixing the location of $x_{\lfloor n/2\rfloor - (k-1)}$, the total unnormalized mass across $I$ is maximized when the datapoints $x_1,x_2,\ldots, x_{\lfloor n/2\rfloor -k}$ are placed as leftmost as possible. Because $D$ must satisfy the collapsing property, this occurs when \[x_i = -\frac m2 + (i-1) \frac{x_{\lfloor n/2\rfloor -(k-1)}-\frac m2}{\lfloor n/2\rfloor -k}\] for $i\in \{1,2,\ldots, \lfloor n/2\rfloor -k\}$, i.e., they are spaced uniformly. Lastly, we note that fixing these locations for $x_1,\ldots, x_{\lfloor n/2\rfloor-k}$, the total unnormalized mass of $p_{\alpha}(D,\ell)>k$ is maximized when $I$ is longest, or when $x_{\lfloor n/2\rfloor-(k-1)}=m/2$.
    
    Let $m_{\le k}$ and $m_{>k}$ be the unnormalized mass on $p_\alpha (D, \Mo (D))\le k$ and $p_\alpha (D, \Mo (D))>k$ respectively. Since \[\Pr\big[p_\alpha (D, \Mo (D))\le k \,\big|\, \M(D) <\T(D)\big]= \frac{m_{\le k}}{m_{\le k} + m_{>k}},\] the probability is both an increasing function of $m_{\le k}$ and a decreasing function of $m_{> k}$. However, as shown via the conditions derived above, we can simultaneously maximize $m_{\le k}$ and minimize $m_{>k}$ by taking $D_k^\ctm$ where $x_i = -m/2 + (i-1) \frac{x_{\lfloor n/2\rfloor -(k-1)}-\frac m2}{\lfloor n/2\rfloor -k}$ for $i\le \lfloor n/2\rfloor -k$ and $x_j = \T(D) = m/2$ for $j> \lfloor n/2\rfloor -k$.
    
    Thus, $D_k^\ctm$ minimizes $\Pr\big[p_\alpha (D, \Mo (D))\le k \,\big|\, \M(D) <\T(D)\big]$ across all $D\in \ctm$, but since $\T(D_k^\ctm) = m/2$ we also have \[\Pr\big[p_\alpha (D_k^\ctm, \Mo (D_k^\ctm))\le k \,\big|\, \M(D_k^\ctm) <\T(D_k^\ctm)\big]=\Pr [p_\alpha (D_k^\ctm, \Mo (D_k^\ctm))\le k]\] so $D_k^\ctm$ must also minimize $\Pr [p_\alpha (D_k^\ctm, \Mo (D_k^\ctm))\le k]$ as required.
\end{proof}

Given the previous lemma, we are now able to deduce a high probability bound for the magnitude of the score $p_\alpha$ on outputs $\MoLong (D)$ over $D\in \ctm$.

\begin{theorem}[High probability upper bound on $p_\alpha$ for $\MoLong$ over $\ctm$]\label{p_bound_mp}
    There exists a universal constant $C$ such that for every $\alpha, \beta \in (0, 1/3), \eps \in (0,1)$, and $n\in \mathbb{N}$ satisfying $n\eps \ge C\cdot \ln (1/(\alpha \beta))$, it holds for every $D\in \ctm$ that \[p_{\alpha} (D,\MoLong (D))=O\left (\frac 1\eps \max \left[1, \ln \left(\frac{1}{\alpha \beta n\eps}\right)\right]\right)\] with probability at least $1-\beta$.
\end{theorem}

 If we set the widening parameter to be $\alpha = 1/(n\eps)$, the bound for $p_\alpha$ simplifies into the following corollary.

 \begin{corollary}[Theorem~\ref{p_bound_mp} with $\alpha = 1/(n\eps)$]\label{p_bound_mp_with_alpha}
     There exists a universal constant $C$ such that for every $\beta \in (0, 1/3), \eps \in (0,1)$, and $n\in \mathbb{N}$ satisfying $n\eps \ge C\cdot \ln (1/\beta)$, it holds for every $D\in \ctm$ that \[p_\alpha (D, \MoLongs (D)) = O\left(\frac 1\eps  \ln \frac 1\beta\right)\] with probability at least $1-\beta$ for widening parameter $\alpha^* = 1/(n\eps)$.
 \end{corollary}

Note that in this result (and subsequent theorems), we ask the failure probability $\beta$ to not be ``unreasonably large" by imposing an upper bound of $1/3$ that appears in Theorem~\ref{p_bound_mp}. Notationally, this allows us to simply remove the $\max$ expression from the final bound. We now provide the proofs of Theorem~\ref{p_bound_mp} and Corollary~\ref{p_bound_mp_with_alpha}.

\begin{proof}[Proof of Theorem~\ref{p_bound_mp}]
    For brevity, denote the $\MoLong$ mechanism by $\Mo$ throughout this proof. By Lemma~\ref{wc_dc}, we have for every fixed $k\in \{0,1,\ldots, \lfloor n/2\rfloor\}$ that \begin{align}\label{first_ineq}
        \Pr [p_{\alpha} (D,\Mo(D))> k] &\le \Pr[p_{\alpha}(D_k^\ctm,\Mo(D_k^\ctm))>k]\\
        &\le 1-\frac{\alpha m}{\alpha m + \frac{m}{\lfloor n/2\rfloor -k}\sum_{i=k+1}^{\lfloor n/2\rfloor} \exp (-\frac \eps 2 i)}
    \end{align} where $D_k^\ctm$ is the dataset from Lemma~\ref{wc_dc} and the second inequality holds because the $\sum_{i=k+1}^{\lfloor n/2\rfloor} \exp (-\frac \eps 2 i)$ term overestimates the unnormalized failure mass. Note that \begin{align*}
        \sum_{i=k+1}^{\lfloor n/2\rfloor} \exp \left(-\frac \eps 2 i\right)=\frac{\exp(-\frac \eps 2 (k+1))(1-\exp(-\frac \eps 2 (\lfloor n/2\rfloor -k)))}{1-\exp(-\frac \eps 2)}
    \end{align*} by the geometric series formula, which applied to \eqref{first_ineq} yields \begin{align*}
        \Pr [p_{\alpha} (D,\Mo(D))> k] &\le \frac{1}{\alpha m} \cdot \frac{m}{\lfloor n/2\rfloor -k} \cdot \frac{\exp(-\frac \eps 2 (k+1))(1-\exp(-\frac \eps 2 (\lfloor n/2\rfloor -k)))}{1-\exp(-\frac \eps 2)}\\
        &\le \frac{1-\exp(-\frac \eps 2 \lfloor n/2\rfloor)}{1-\exp(-\frac \eps 2)}\cdot \frac{1}{\alpha (\lfloor n/2\rfloor -k)}\exp\left(-\frac \eps 2 (k+1)\right)
    \end{align*} where the first step bounds the former denominator by $\alpha m$ and the second step bounds $\exp\left(-\frac \eps 2\lfloor n/2\rfloor-k\right) $ by $\exp\left(-\frac \eps 2\lfloor n/2\rfloor\right) $. Therefore, we are guaranteed to have $\Pr [p_{\alpha} (D,\Mo(D))> k]\le \beta$ for any $k\in \{0,1,\ldots, \lfloor n/2\rfloor\}$ on which \[\frac{1-\exp(-\frac \eps 2 \lfloor n/2\rfloor)}{1-\exp(-\frac \eps 2)}\cdot \frac{1}{\alpha (\lfloor n/2\rfloor -k)}\exp(-\frac \eps 2 (k+1))\le \beta\] or equivalently, \begin{equation}\label{size_of_k}
        k\ge -1+\frac 2\eps\ln \left(\frac{1}{\alpha\beta (\lfloor n/2\rfloor -k)}\cdot \frac{1-\exp(-\frac \eps 2 \lfloor n/2\rfloor)}{1-\exp(-\frac \eps 2)}\right).
    \end{equation} If $k=0$ makes the right side of \eqref{size_of_k} negative, then the desired theorem result is immediate. Otherwise, the right side must be nonnegative, in which case consider the value \[k^* = \left\lceil -1+\frac{2}{\eps}\ln \left(\frac{2}{\alpha \beta\lfloor n/2\rfloor}\cdot \frac{1-\exp(-\frac \eps 2 \lfloor n/2\rfloor)}{1-\exp(-\frac \eps 2)}\right)\right\rceil,\] which must also be nonnegative. We claim that $k^*$ satisfies Inequality~\eqref{size_of_k}. Note that $\frac{1-\exp(-\frac \eps 2 \lfloor n/2\rfloor)}{1-\exp(-\frac \eps 2)}$ strictly increases as $\epsilon \to 0$ and \[\lim_{\eps\to 0^+} \frac{1-\exp(-\frac \eps 2 \lfloor n/2\rfloor)}{1-\exp(-\frac \eps 2)}=\lfloor n/2\rfloor,\] so \[k^*\le \left\lceil -1+\frac 2\eps \ln \left(\frac{2}{\alpha \beta\lfloor n/2\rfloor}\cdot \lfloor n/2\rfloor\right)\right\rceil = \left\lceil -1+\frac 2\eps \ln \left(\frac{2}{\alpha \beta}\right)\right\rceil.\] This then implies that there must exist a constant $C'\ge 2$ for which $\lfloor n/2\rfloor \ge C'\cdot k^*$ if we choose a sufficiently large constant $C$ such that $n\eps \ge C\cdot \ln (1/(\alpha\beta))$. Thus, noting that $C'/(C'-1)\le 2$, we have that \begin{align*}
        k^* &\ge -1+\frac 2\eps \ln \left(\frac{1}{\alpha \beta \left(1-\frac{1}{C'}\right)\lfloor n/2\rfloor}\cdot \frac{1-\exp(-\frac \eps 2 \lfloor n/2\rfloor)}{1-\exp(-\frac \eps 2)}\right) \\
        &\ge -1+\frac 2\eps \ln \left(\frac{1}{\alpha\beta (\lfloor n/2\rfloor -k^*)}\cdot \frac{1-\exp(-\frac \eps 2 \lfloor n/2\rfloor)}{1-\exp(-\frac \eps 2)}\right)
    \end{align*}
    so indeed $k^*$ satisfies Inequality~\eqref{size_of_k}. Therefore, with failure probability at most $\beta$, we have \begin{align*}
        p_\alpha(D,\Mo(D)) &\le \max (0, k^*)\\
        &= O\left (\frac 1\eps \max \left[1, \ln \left(\frac{1}{\alpha \beta n}\cdot \frac{1-\exp(-\frac \eps 2 \lfloor n/2\rfloor)}{1-\exp(-\frac \eps 2)}\right)\right]\right)\\
        &= O\left(\frac 1\eps \max\left[1,\ln \frac{1}{\alpha \beta n \eps}\right]\right)
    \end{align*} since $\frac{(1-\exp(-\frac \eps 2 \lfloor n/2\rfloor))}{1-\exp(-\frac \eps 2)}=O(1/\eps)$ when $\epsilon \le 1$ and $n\eps \gg 1$.
\end{proof}

\begin{proof}[Proof of Corollary~\ref{p_bound_mp_with_alpha}]
There exists some constant $C$ under which Theorem~\ref{p_bound_mp} applies to any $\alpha,\beta\in (0,1/3), \eps \in (0,1)$ and $n\in \mathbb{N}$ satisfying $n\eps \ge C\cdot \ln (1/(\alpha\beta))$. When $\alpha = 1/(n\eps)$, we have \[n\eps / \ln (1/\alpha) = n\eps/\ln (n\eps)\to\infty\] as $n\eps \to \infty$, so there must exist another constant $C'$ where, for all $n$ and $\eps$ such that $n\eps \ge C'\cdot \ln (1/\beta)$, we have that $n\eps \ge C \cdot (\ln (n\eps)+\ln (1/\beta)) = C\cdot \ln (n\eps/\beta)$. In other words, Theorem~\ref{p_bound_mp} is always applicable when $\alpha = 1/(n\eps)$, provided that $n\eps$ is sufficiently large.

Thus, for $n$ and $\eps\in O(1)$ where $n\eps $ is sufficiently large, it holds with probability at least $1-\beta$ for every $D\in \ctm$ that \begin{align*}
    p_{\alpha} (D,\MoLong (D)) =O\left(\frac 1\eps \max\left[1,\ln \frac{n\eps}{ \beta n \eps}\right]\right) = O\left(\frac 1\eps \ln \frac 1\beta\right)
\end{align*} since $\beta < 1/3<1/e$.
\end{proof}

The high probability bound on $p_\alpha$ consequently informs a similar bound on $\fair$ upon bounding $\fair$ by a function of $p_{\alpha}$. This corresponds to Step~\ref{procedure_step_2} of Procedure~\ref{procedure}. We can relate $p_{\alpha}$, a count of points, to $\fair$, a measure of distance between points.

\begin{lemma}\label{dc_unif}
    For every $D\in \ctm$, it holds that \[|x_{\lceil n/2\rceil} - x_{\lceil n/2\rceil + j}|\le |j|\cdot \frac{m}{\lceil n/2\rceil -1}\] for any $j\in \{-\lceil n/2\rceil +1, \ldots, 0,\ldots, n-\lceil n/2\rceil\}$.
\end{lemma}

This lemma puts a limit on ``how far'' a point $x_i$ can get from $\T(D)$, in terms of the number of points between $x_i$ and $\T(D)$. It also implies that a dataset with $x_{j}=m/2$ for $j\ge \lceil n/2\rceil$ and the rest of the $x_i$s evenly spaced over the entire domain will be the one with agents ``spread furthest from $\T(D)$.'' 

\begin{proof}
    Assume without loss of generality that $j$ is negative; the proof is analogous if $j$ is nonnegative. Since $D$ is collapsing towards the median, it follows that \[x_{i+1}-x_i \ge \frac {1}{|j|} (x_{\lceil n/2\rceil} - x_{\lceil n/2\rceil -|j|})\] for any $i<\lfloor n/2\rfloor -|j|$, which means \begin{align*}
        m \ge x_{\lceil n/2\rceil} - x_1 &=(x_{\lceil n/2\rceil} -x_{\lceil n/2\rceil -|j|})+\sum_{i=1}^{\lceil n/2\rceil -|j| -1} (x_{i+1}-x_i) \\
        &\ge (x_{\lceil n/2\rceil} -x_{\lceil n/2\rceil -|j|}) + \frac{\lceil n/2\rceil -|j|-1}{|j|} (x_{\lceil n/2\rceil} -x_{\lceil n/2\rceil -|j|})\\
        &= \frac{\lceil n/2\rceil-1}{|j|}(x_{\lceil n/2\rceil} -x_{\lceil n/2\rceil -|j|})
    \end{align*} and rearranging immediately gives the result.
\end{proof}

We can now prove an upper bound on $\fair$.

    \begin{theorem}[High probability bound on $\fair$ for $\MoLongs$ over $\ctm$]\label{fair_bound_mp}
    There exists a universal constant $C$ such that for every $\beta \in (0, 1/3), \eps \in (0,1)$, and $n\in \mathbb{N}$ satisfying $n\eps \ge C\cdot \ln (1/\beta)$, it holds for every $D\in \ctm$ that \[\fair (D,\MoLongs (D))=O\left(\frac{m}{n\eps} \ln \frac 1\beta\right)\] with probability at least $1-\beta$ for widening parameter $\alpha^*=1/(n\eps)$.
\end{theorem}

\begin{proof}
    First, note that for any fixed value of $p_\alpha(D,\ell)< \lceil n/2\rceil$, we have that \begin{equation}\label{fair_ineq}
\fair (D, \ell) \le x_{\lceil n/2\rceil} - x_{\lceil n/2\rceil - p_\alpha(D,\ell)}+\alpha m \le p_\alpha(D,\ell)\frac{m}{\lceil n/2\rceil -1}+\alpha m
\end{equation}

    where the second inequality follows by Lemma~\ref{dc_unif}. In Theorem~\ref{p_bound_mp}, we derived that \[p_\alpha(D,\MoLong (D))\le \max\left(0, \left\lceil -1+\frac{2}{\eps}\ln \left(\frac{2}{\alpha \beta\lfloor n/2\rfloor}\cdot \frac{1-\exp(-\frac \eps 2 \lfloor n/2\rfloor)}{1-\exp(-\frac \eps 2)}\right)\right\rceil\right)\] with probability at least $1-\beta$ for a general widening parameter $\alpha$, as long as $n\eps$ is sufficiently large compared to $\ln (1/\alpha)$. Assuming this is still the case, it then also holds that \[\max\left(0, \left\lceil -1+\frac{2}{\eps}\ln \left(\frac{2}{\alpha \beta\lfloor n/2\rfloor}\cdot \frac{1-\exp(-\frac \eps 2 \lfloor n/2\rfloor)}{1-\exp(-\frac \eps 2)}\right)\right\rceil\right)<\lfloor n/2\rfloor,\] which means that when $\MoLong$ does not ``fail'' over $p_\alpha$, the relation~\eqref{fair_ineq} can be applied and says that \begin{align}
        \fair (D, \MoLong (D)) &\le p_{\alpha} (D,\MoLong (D)) \frac{m}{\lfloor n/2\rfloor -1} + \alpha m\\
        &= \frac{m}{\lfloor n/2\rfloor -1} \left [O\left (\frac 1\eps \max \left[1, \ln \left(\frac{1}{\alpha \beta n}\cdot \frac{1-\exp(-\frac \eps 2 \lfloor n/2\rfloor)}{1-\exp(-\frac \eps 2)}\right)\right]\right)\right] + \alpha m\\
        &= O\left (\frac {m}{n\eps} \max \left[1, \ln \left(\frac{1}{\alpha \beta n\eps}\right)\right]+\alpha m\right)\label{whp_fair_dc}
    \end{align}

    From~\eqref{whp_fair_dc}, it becomes apparent that the optimal $\alpha$ is $\alpha = \alpha^* = 1/(n\eps)$. Applying Corollary~\ref{p_bound_mp_with_alpha} yields that with probability at least $1-\beta$, \begin{align*}
        \fair(D, \MoLongs (D)) = O\left(\frac{m}{n\eps}\ln \frac 1\beta\right)
    \end{align*} as claimed.
\end{proof}

Lastly, as suggested by Step~\ref{procedure_step_3} from Procedure~\ref{procedure}, we can move from the bounds on $p_\alpha$ and $\fair$ to a bound on $\swdiff$ by leveraging the relationship between the three quantities.

    \begin{theorem}[High probability bound on $\swdiff$ for $\MoLongs $ over $\ctm$]\label{swdiff_bound_mp}
    There exists a universal constant $C$ such that for every $\beta \in (0, 1/3), \eps \in (0,1)$, and $n\in \mathbb{N}$ satisfying $n\eps \ge C\cdot \ln (1/\beta)$, it holds for every $D\in \ctm$ that \[\swdiff (D,\MoLongs (D))=O\left(\frac{m}{n\eps^2} \left(\ln \frac 1\beta\right)^2\right)\] with probability at least $1-\beta$ for widening parameter $\alpha^* = 1/(n\eps)$.
    \end{theorem}

\begin{proof}
    Fix a location $\ell\in V$. Then, the crossed set of agents as defined in Theorem~\ref{swdiff_closed_form} has cardinality (at most) $p_{\alpha}(D,\ell)-1$, and the distance between any crossed agent and $\ell$ is at most $d(\T(D), \ell)=\fair (D,\ell)$. Therefore, by Theorem~\ref{swdiff_closed_form}, \begin{equation}\label{swdiff_fair_p_ineq}\swdiff(D,\ell) \le (2p_{\alpha}(D,\ell)-1)\cdot \fair (D,\ell).\end{equation} Taking $\alpha = 1/(n\eps)$ and substituting into \eqref{swdiff_fair_p_ineq} the high probability bounds from Corollary~\ref{p_bound_mp_with_alpha} and Theorem~\ref{fair_bound_mp} implies that with probability $1-\beta$, $\swdiff(D, \MoLongs(D))$ is \[O\left(\frac 1\eps \ln \frac 1\beta\right)\cdot O\left(\frac{m}{n\eps}\max \ln \frac 1\beta\right)=O\left(\frac{m}{n\eps^2}\left(\ln \frac 1\beta\right)^2\right).\]
\end{proof}

\subsubsection{Upper bounds over $\spm_\lambda$}

We now turn towards upper bounds for $\MoLong$ over our other family $\spm_\lambda$ and follow Procedure~\ref{procedure} once more. The proofs are similar, but the additional freedom granted to datasets in $\spm_\lambda$ will introduce new $\lambda$ terms.

Once more, we first characterize the ``worst-case'' datasets in $\spm_\lambda$ under which it is most likely that $p_{\alpha} (D,\MoLong (D))$ is large. Compared to the worst-case datasets in $\ctm$ described by Lemma~\ref{wc_dc}, the ones in $\spm_\lambda$ will have additional agents stacked on the endpoint of $V$ opposite of the median, which results in sparser agent placement throughout $V$.

\begin{lemma}\label{wc_ddelta}
    Fix $\lambda\ge 0$ and let $s=\lceil \lambda n\rceil -1$. For every fixed $k\in \{0,1,\ldots, \lfloor n/2\rfloor\}$, the probability $\Pr[p_{\alpha}(D, \MoLong (D)\le k]$ obtains its minimum across $D\in \spm_\lambda$ on the dataset $D_k^\spm$ with $\lceil n/2\rceil +k$ points at $m/2$, $s$ points at $-m/2$, and the remaining $\lfloor n/2\rfloor -k-s$ points spaced at $-m/2 + i\frac{m}{\lfloor n/2\rfloor -k+\lambda n}$ for $i=\{1,2,\ldots, \lfloor n/2\rfloor -k-s\}$.
\end{lemma}

\begin{figure}[h]
    \centering
    \begin{tikzpicture}

        \draw[thick] (-6,0) -- (6,0) node[midway, above] {};
        
        \draw[solid, thick] (-6,-0.2) -- (-6,0.2) node[below=10] {$-m/2$};
        \draw[solid, thick] (-3,-0.1) -- (-3,0.1) node[below=7] {$-m/4$};
        \draw[solid, thick] (0,-0.1) -- (0,0.1) node[below=7] {$\vphantom{/}0$};
        \draw[solid, thick] (3,-0.1) -- (3,0.1) node[below=7] {$m/4$};
        \draw[solid, thick] (6,-0.2) -- (6,0.2) node[below=10] {$m/2$};
        
        \foreach \y in {0.5,0.8,1.1} {
            \filldraw[blue] (6,\y) circle (2pt);
        }
        \node at (6,1.6) {\vdots};
        \foreach \y in {1.9,2.2,2.5} {
            \filldraw[blue] (6,\y) circle (2pt);
        }
        
        \foreach \y in {0.5,0.8} {
            \filldraw[blue] (-6,\y) circle (2pt);
        }
        \node at (-6,1.35) {\vdots};
        \foreach \y in {1.7} {
            \filldraw[blue] (-6,\y) circle (2pt);
        }
        \foreach \x in {-6,-4.8,-3.6,-2.4,-1.2,1.2,2.4} {
            \filldraw[blue] (\x, 0.5) circle (2pt);
        }

        \node at (0, 0.5) {\textbf{$\cdots$}};

        \draw [thick, decorate,decoration={brace,mirror,amplitude=5pt}] (-6.4,1.8) -- (-6.4,0.4) node[midway,xshift=-12pt] {$s$};
        \draw [thick,decorate,decoration={brace,amplitude=5pt}] (6.4,2.6) -- (6.4,0.4) node[midway,xshift=34pt] {$\lceil n/2 \rceil + k$};
        
        \node[blue] at (6, -1.6) {$\mathcal{T}(D_k^\spm)$};
        \draw[blue, thick,->] (6,-1.3) -- (6,-0.9);
    \end{tikzpicture}
    \caption{Depiction of Dataset $D_k^\spm$, a worst-case dataset in $\spm_\lambda$ for Lemma~\ref{wc_ddelta}. Blue points indicate agent locations along $V$. Compared to the worst-case dataset in $\ctm$ for Lemma~\ref{wc_dc}, this dataset has agents spread more sparsely and an additional stack of $s$ agents away from the median.}
    \label{fig:wc_dataset_ddelta}
\end{figure}
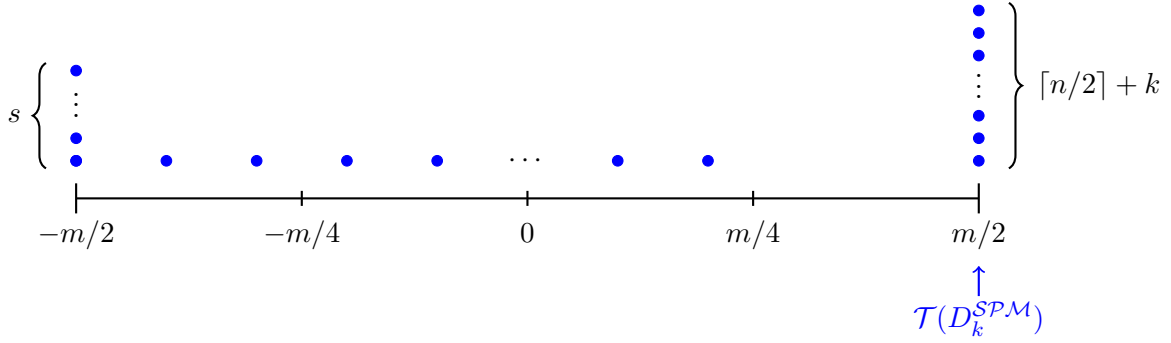

\begin{proof}
    For brevity, denote the $\MoLong$ mechanism by $\Mo$ throughout this proof. The proof follows a structure similar to that of Lemma~\ref{wc_dc}. If $L_D$ and $R_D$ again represent the events $\M(D)<\T(D)$ and $\M(D)>\T(D)$, then once more \begin{align*}
    \Pr[p_\alpha (D, \Mo (D))\le k] &= \Pr[L_D] 
    \cdot \Pr\big[p_\alpha (D, \Mo (D))\le k \,\big|\, L_D\big] \\
    &\quad + \Pr[R_D] 
    \cdot \Pr\big[p_\alpha (D, \Mo (D))\le k \,\big|\, R_D\big] \\
    &\ge \min \{\Pr\big[p_\alpha (D, \Mo (D))\le k \,\big|\, L_D\big], \Pr\big[p_\alpha (D, \Mo (D))\le k \,\big|\, R_D\big]\}.
\end{align*} Therefore, we can simply find what dataset minimizes $\Pr\big[p_\alpha (D, \Mo (D))\le k \,\big|\, L_D\big]$ and show that $\Pr\big[p_\alpha (D, \Mo (D))\le k \,\big|\, L_D\big]=\Pr[p_\alpha (D, \Mo (D))\le k]$ on this minimizing dataset. Accordingly, assume henceforth that $\Mo (D)<\T(D)$ and we already condition all events on this moving forward.

For any $D\in \spm_\lambda$ where $p_{\alpha} (D,\Mo (D))>k$ has nonzero probability mass, the interval $(\T(D)-\alpha m, \T(D))$ of length $\alpha m$ always contributes at least \[\alpha m \cdot \exp(-\eps/2\cdot 0) = \alpha m\] in unnormalized mass under $\Mo $'s weighting to the event $p_{\alpha} (D,\Mo (D))\le k$. Equality holds when $x_{\lfloor n/2\rfloor -(k-1)}=\T(D)$. 

Meanwhile, $I = [-m/2, x_{\lfloor n/2\rfloor -(k-1)}-\alpha m]$ is the set of points $\ell$ for which $p_{\alpha}(D,\ell)>k$. As in Lemma~\ref{wc_dc}, fix $x_{\lfloor n/2\rfloor - (k-1)}$ for now. The total unnormalized mass across $I$ is maximized when the points $x_1,\ldots, x_{\lfloor n/2\rfloor -k}$ are as leftmost as possible. To characterize this placement, we first consider the ``true'' distribution $P\in \mathcal{P}$ that must be at most $\lambda$ away from $D$ in K-S distance. Since $D$ has density of $(\lceil n/2\rceil +k)/n$ over $[x_{\lfloor n/2\rfloor -(k-1)}, m/2]$, $P$ can have at most \[1-\left(\frac{\lceil n/2\rceil +k}{n} -\lambda\right)=\frac{\lfloor n/2\rfloor -k}{n}+\lambda\] density over $[-m/2, x_{\lfloor n/2\rfloor -(k-1)}]$. Therefore, the single-peaked distribution $P$ that enables the leftmost arrangement of points $x_1,\ldots, x_{\lfloor n/2\rfloor -k}$ must spread the density of $(\lfloor n/2\rfloor -k)/n+\lambda$ uniformly from $[-m/2, x_{\lfloor n/2\rfloor -(k-1)}]$.

Given this formulation of $P$, the leftmost arrangement of $x_1,\ldots, x_{\lfloor n/2\rfloor -k}$ can be done greedily. Let $m'=x_{\lfloor n/2\rfloor -(k-1)}-(-m/2)$ be the length of the eligible interval. First, since $s$ is the largest integer for which $\lambda > s/n$, we safely place at most $s$ points at $-m/2$ while staying within $\lambda$ away in K-S distance from $P$. The placement of subsequent points now must track the speed at which density accumulates for $P$. In particular, it takes $P$ a segment of length \[\frac{1/n}{(1/m') \cdot ((\lfloor n/2\rfloor -k)/n+\lambda)}=\frac{m'}{\lfloor n/2\rfloor -k+\lambda n}\] to accrue a cumulative density of $1/n$, which means we can only place a new point $x_i$ every $m'/(\lfloor n/2\rfloor -k+\lambda n)$ in distance after the previous. Hence, we finish with $s$ points at $-m/2$ and $\lfloor n/2\rfloor -k-s$ points spread every $m'/(\lfloor n/2\rfloor -k+\lambda n)$ thereafter. This corresponds to the arrangement in Figure~\ref{fig:wc_dataset_ddelta}. Note that these subsequent points do not necessarily evenly space the entire interval $[-m/2, x_{\lfloor n/2\rfloor -(k-1)}]$, but rather skew leftwards.

We now formally verify that this arrangement of points obeys the K-S distance requirement. Let $F$ and $F_n$ be the density functions of $P$ and $D$ respectively, and note that $F(\ell)-F_n(\ell)$ is maximized on $\ell\in [-m/2, x_{\lfloor n/2\rfloor -(k-1)}]$ at $\ell =x_{\lfloor n/2\rfloor -(k-1)}$, so $F(\ell)-F_n(\ell)<\lambda $ always holds. Meanwhile, $F_n(\ell)-F(\ell)$ is maximized on $\ell=x_{s+i}$ for any $i\in \{1, \ldots, \lfloor n/2\rfloor -k - s\}$, and at those points, \[F(x_{s+i}) = i\frac{m'}{\lfloor n/2\rfloor -k+\lambda n}\left(\frac{\lfloor n/2\rfloor -k+\lambda n}{m'n}\right)=\frac in,\] \[F_n(x_{s+i}) = \frac sn+\frac in\] so $F_n(\ell) - F(\ell)\le s/n <\lambda$ always. Therefore, $D$ and $P$ are within $\lambda$ by K-S distance, assuming suitable placements of $x_{\lfloor n/2\rfloor -(k-1)},\ldots, x_{n}$ relative to $P$. This fixes the maximizing arrangement of $x_1,\ldots, x_{\lfloor n/2\rfloor -k}$ given $x_{\lfloor n/2\rfloor -(k-1)}$'s location. Therefore, the total unnormalized mass attributed to $I$ will be maximized when this arrangement is used and $x_{\lfloor n/2\rfloor -(k-1)}=m/2$.

We can simultaneously minimize the unnormalized mass on $p_\alpha (D,\Mo (D))\le k$ and maximize the unnormalized mass on $p_\alpha (D,\Mo (D))>k$ by selecting $x_{\lfloor n/2\rfloor -(k-1)}=\T(D) = m/2$ and arranging $x_1,\ldots, x_{\lfloor n/2\rfloor -k}$ as specified previously. This also forces $x_{\lfloor n/2\rfloor -(k-1)}, \ldots, x_n$ to all be at $m/2$. Therefore, this collection of points, which describes $D_k^\spm$ exactly, minimizes the probability $\Pr\big[p_\alpha (D, \Mo (D))\le k \,\big|\, L_D\big]$ across $D\in \spm_\lambda$. Moreover, \[\Pr\big[p_\alpha (D_k^\spm, \Mo (D_k^\spm))\le k \,\big|\, L_{D_k^\spm}\big]=\Pr[p_\alpha (D_k^\spm, \Mo (D_k^\spm))\le k]\] so $D_k^\spm$ also minimizes $\Pr[p_\alpha (D, \Mo (D))\le k]$ across $D\in \spm_\lambda$ by our initial reasoning. This completes the proof.
\end{proof}

We may now obtain a high probability upper bound on $p_\alpha$ again, where we note that the result takes very similar form to Theorem~\ref{p_bound_mp}.

\begin{theorem}[High probability upper bound on $p_\alpha$ for $\MoLong$ over $\spm_\lambda$]\label{p_bound_mp_ddelta}
    There exists a universal constant $C$ such that for every $\beta \in (0, 1/3), \eps \in (0,1)$, and $n\in \mathbb{N}$ satisfying $n\eps \ge C\cdot \ln (1/(\alpha\beta))$, it holds for every $D\in \spm_\lambda$ that \[p_{\alpha} (D,\MoLong (D))=O\left(\frac 1\eps \max \left [1, \ln \left(\frac{\lambda}{\alpha \beta n}\right)\right]\right)\] with probability at least $1-\beta$.
\end{theorem}

As in Theorem~\ref{p_bound_mp}, this theorem holds for general $\alpha$ but we will later optimize over $\alpha$ to produce the tightest possible result for $\fair$ and $\swdiff$. Like before, it turns out we will select $\alpha = 1/(n\eps)$, upon which Theorem~\ref{p_bound_mp_ddelta} simplifies to the following corollary.

 \begin{corollary}[Theorem~\ref{p_bound_mp} with $\alpha = 1/(n\eps)$]\label{p_bound_mp_ddelta_with_alpha}
     There exists a universal constant $C$ such that for every $\beta \in (0, 1/3), \eps \in (0,1)$, and $n\in \mathbb{N}$ satisfying $n\eps \ge C\cdot \ln (1/\beta)$ and $n\ge \beta/\lambda$, it holds for every $D\in \spm_\lambda$ that \[p_\alpha (D, \MoLongs (D)) = O\left(\frac 1\eps \ln \frac {n\lambda}{\beta}\right)\] with probability at least $1-\beta$ for widening parameter $\alpha^* = 1/(n\eps)$.
 \end{corollary} Observe that the $n\ge \beta / \lambda$ condition is necessary for technical precision, but does not impose any ``new" constraint since we generally assume $\lambda = \Theta (1/\sqrt{n})$.

 We now offer the proofs of Theorem~\ref{p_bound_mp_ddelta} and Corollary~\ref{p_bound_mp_ddelta_with_alpha}.

\begin{proof}[Proof of Theorem~\ref{p_bound_mp_ddelta}]
    The proof is similar to the one for Theorem~\ref{p_bound_mp}. For brevity, denote the $\MoLong$ mechanism by $\Mo$ throughout. Let $s$ be the largest integer for which $\lambda > s/n$. For every fixed $k\in \{0,1\ldots, \lfloor n/2\rfloor\}$, applying the exponential mechanism weighting to the dataset $D_k^\spm$ from Lemma~\ref{wc_ddelta} will result in an unnormalized mass of $\alpha m$ for the event $p_{\alpha}(D_k^\spm,\Mo (D_k^\spm))\le k$. Meanwhile, an overestimate on the unnormalized mass $M$ for the entire outcome space is given by \[M=\alpha m+\frac{m}{\lfloor n/2\rfloor -k+n\lambda} \sum_{i=k+1}^{\lfloor n/2\rfloor -s + 1}\exp\left(-\frac \eps 2 i\right) +\left(m-\frac{m(\lfloor n/2\rfloor -s-k)}{\lfloor n/2\rfloor -k+n\lambda}\right)\exp\left(-\frac \eps 2 (k+1)\right),\] which simplifies via the geometric series formula to be \begin{align*}
    M &=\alpha m +\left(\frac{m(1-\exp(-\frac \eps 2(\lfloor n/2\rfloor -s-k+1)))}{(\lfloor n/2\rfloor-k+n\lambda)(1-\exp(-\frac \eps 2))}+\frac{m(n\lambda+s)}{\lfloor n/2\rfloor -k+n\lambda}\right)\exp\left(-\frac \eps 2 (k+1)\right)\\
    &\le \alpha m + m\cdot \frac{2n\lambda}{\lfloor n/2\rfloor -k}\cdot \frac{1-\exp(-\frac \eps 2\lfloor n/2\rfloor)}{1-\exp(-\frac \eps 2)}\exp\left(-\frac \eps 2 (k+1)\right).\end{align*} Therefore, we can establish the bound \begin{align*}
        \Pr[p_{\alpha}(D_k^\spm,\Mo (D_k^\spm))> k]  &= 1- \Pr[p_{\alpha}(D_k^\spm,\Mo (D_k^\spm))\le k] \\
        &\le 1-\frac{\alpha m}{M}\\
        &\le \frac{M-\alpha m}{\alpha m}\\
        &\le \frac{1-\exp(-\frac \eps 2 \lfloor n/2\rfloor)}{1-\exp(-\frac \eps 2)}\cdot \frac{1}{\alpha (\lfloor n/2\rfloor -k)}\exp\left(-\frac \eps 2 (k+1)\right)
    \end{align*}
  from which we are assured \[\Pr[p_\alpha (D,\Mo (D))> k]\le \beta\] for all $k\in \{0,1,\ldots, \lfloor n/2\rfloor\}$ that satisfy \[\frac{1-\exp(-\frac \eps 2 \lfloor n/2\rfloor)}{1-\exp(-\frac \eps 2)}\cdot \frac{1}{\alpha (\lfloor n/2\rfloor -k)}\exp\left(-\frac \eps 2 (k+1)\right)\le \beta,\]or equivalently, \begin{equation} \label{ddelta_k_size}
        k\ge -1+\frac{2}{\eps} \ln \left (\frac{2n\lambda}{\alpha\beta(\lfloor n/2\rfloor -k)}\cdot \frac{1-\exp(-\frac \eps 2\lfloor n/2\rfloor)}{1-\exp(-\frac \eps 2)}\right).
    \end{equation}

    As in Theorem~\ref{p_bound_mp}, if $k=0$ makes the right hand side negative, then we immediately obtain the desired result. Otherwise, if this is not the case, we can consider the nonnegative quantity \[k^* = \left\lceil-1+\frac 2\eps \ln \left(\frac{4n\lambda}{\alpha\beta\lfloor n/2\rfloor}\cdot \frac{1-\exp(-\frac \eps 2\lfloor n/2\rfloor)}{1-\exp(-\frac \eps 2)}\right)\right\rceil,\] which we claim satisfies Inequality \eqref{ddelta_k_size}. Again noting that \[\frac{1-\exp(-\frac \eps 2\lfloor n/2\rfloor)}{1-\exp(-\frac \eps 2)}\le \lfloor n/2\rfloor\] implies \[k^*\le \left\lceil -1+\frac 2\eps \ln \left(\frac{4n\lambda}{\alpha\beta}\right)\right\rceil.\] Therefore, there must exist a constant $C'\ge 2$ for which $\lfloor n/2\rfloor \ge C'\cdot k^*$ if we choose a sufficiently large constant $C$ such that $n\eps \ge C\cdot \ln (1/(\alpha\beta))$. Hence, we can use identical reasoning as in Theorem~\ref{p_bound_mp} to conclude that $k^*$ fulfills Inequality~\eqref{ddelta_k_size}. It thus holds that with probability at least $1-\beta$, we have \begin{align*} p_\alpha (D,\Mo (D)) \le \max (0, k^*) &= O\left (\frac 1\eps \max \left[1, \ln \left(\frac{\lambda}{\alpha\beta}\cdot \frac{1-\exp(-\frac \eps 2\lfloor n/2\rfloor)}{1-\exp(-\frac \eps 2)}\right)\right]\right)\\
    &= O\left(\frac 1\eps \max \left [1, \ln \left(\frac{\lambda}{\alpha \beta n}\right)\right]\right)\end{align*} by noting that $\frac{1-\exp(-\frac \eps 2\lfloor n/2\rfloor)}{1-\exp(-\frac \eps 2)}=O(1/\eps)$ as before.
\end{proof}

\begin{proof}[Proof of Corollary~\ref{p_bound_mp_ddelta_with_alpha}]
This proof is nearly identical to the one for Corollary~\ref{p_bound_mp_with_alpha}, so we refer the reader there.

\end{proof}

We now transfer the bound from $p_{\alpha}$ over to $\fair$ by relating the two loss functions. Illustrating the parallels $\spm_\lambda$ and $\mathcal{P}$ have with $\ctm$, we can produce a result for $\mathcal{P}$ similar in spirit to Lemma~\ref{dc_unif}.

\begin{lemma}\label{ddelta_collapse}
    Let $F$ be the CDF for some distribution $P\in \mathcal{P}$. Then, for every $\ell\in V$, it holds that \[|F^{-1}(0.5) -\ell|\le 2m\cdot |0.5 - F(\ell)|.\]
\end{lemma}

Intuitively, this lemma formalizes the idea that the most spread out (with respect to its median) distribution $P\in \mathcal{P}$ is the uniform distribution.

\begin{proof}
    Suppose for the sake of contradiction that for some $P$, it holds that $|F^{-1}(0.5)-\ell|> 2m|0.5-F(\ell)|$ for some $\ell\in V$. Assume without loss of generality that $\ell < F^{-1}(0.5)$. The reverse case has an analogous proof. Since $f(a) \ge f(\ell)$ for all $a\in [\ell, F^{-1}(0.5)]$ by single-peakedness at $F^{-1}(0.5)$, it follows that \[f(\ell) \le \frac{1}{F^{-1}(0.5)-\ell}|0.5 - F(\ell)| < \frac{1}{F^{-1}(0.5)-\ell}\left(\frac {1}{2m} |F^{-1}(0.5)-\ell|\right) = \frac {1}{2m}.\] Since furthermore $f(a)\le f(\ell)$ for all $a<\ell$, it then holds that \[F(\ell) - F(-m/2) < (\ell - (-m/2)) \cdot \frac {1}{2m} \] and so \[0.5 = F(F^{-1}(0.5)) - F(-m/2) < \frac{F^{-1}(0.5)-\ell}{2m} + \frac{\ell - (-m/2)}{2m}< \frac{m}{2m} = 0.5,\] yielding a contradiction. Hence, the claimed statement holds.
\end{proof}

This result now allows us to relate density of agents, represented by expressions involving $F$, to distances over $\mathcal{P}$, represented by expressions involving $F^{-1}$. This gives us a mapping between the bound on $p_\alpha$ to one on $\fair$.

    \begin{theorem}[High probability bound on $\fair$ for $\MoLongs$ over $\spm_\lambda$]\label{fair_bound_mp_ddelta}
    There exists a universal constant $C$ such that for every $\beta \in (0, 1/3), \eps \in (0,1)$, and $n\in \mathbb{N}$ satisfying $n\eps \ge C\cdot \ln (1/\beta)$ and $n\ge \beta/\lambda$, it holds for every $D\in \spm_\lambda$ that \[\fair (D,\MoLongs (D))=O\left(\frac{m}{n\eps}\ln \frac{n\lambda}{\beta}+m\lambda\right)\] with probability at least $1-\beta$ for widening parameter $\alpha^* = 1/(n\eps)$.
\end{theorem}

\begin{proof}
Let $P$ be a distribution in $\mathcal{P}$ that $D$ is $\lambda$ away from in K-S distance, and let $F, F_n$ be the CDFs of $P$ and $D$ respectively. Set $F^{-1}(0.5) = t$ and $F^{-1}_n(0.5) =t_n$ with $\eta = |t_n-t|$. Finally, let $a = \mathrm{argmin}_{a: |a-\ell|\le \alpha m} q(D, a)$. We now have \begin{align*}
        \fair (D,\ell) = |t_n-\ell| &\le |t_n-a| + \alpha m\\
        &\le |t-a| + \eta + \alpha m\\
        &\le m|0.5 - F(a)|+\eta+\alpha m\\
        &\le m\left(|F_n(t_n) - F(a)| + \frac 1n\right) + \eta + \alpha m\\
        &\le m\left(|F_n(t_n) - F_n(a)| + \frac 1n +\lambda\right) + \eta + \alpha m\\
        &\le \frac mn (p(D,\ell)+1) + m\lambda + \eta+\alpha m
    \end{align*}
where we utilize Lemma~\ref{ddelta_collapse}. We can bound $\eta$ via another application of Lemma~\ref{ddelta_collapse}. We have that \begin{align*}
    \eta = |t-t_n| &\le 2m|F(t) - F(t_n)|\\
    &\le 2m (|F(t)-F_n(t_n)|+|F_n(t_n)-F(t_n)|)\\
    &\le 2m\left(\frac 1n +\lambda\right)
\end{align*} so the previous inequality becomes \[\fair (D,\ell) \le \frac{m}{n}(p(D,\ell)+1) + \frac{2m}{n}+3m\lambda +\alpha m.\]

Then, via the result of Theorem~\ref{p_bound_mp_ddelta}, we have that $\fair(D,\ell)$ is at most \begin{equation}\label{eq:whp_fair_ddelta}\frac{m}{n}\cdot O\left(\frac 1\eps \max \left [1, \ln \left(\frac{\lambda}{\alpha \beta n}\right)\right]\right) + 3m\lambda+\alpha m\end{equation} with probability $1-\beta$ for a general widening parameter $\alpha$ if $n\eps$ is sufficiently larger than $\ln (1/\alpha)$. As discussed in Corollary~\ref{p_bound_mp_ddelta_with_alpha}, we can select $\alpha =\alpha^*= 1/(n\eps)$ and apply Theorem~\ref{p_bound_mp_ddelta} if $n\eps$ is bigger than a sufficiently large constant. Combining Corollary~\ref{p_bound_mp_ddelta_with_alpha} with \eqref{eq:whp_fair_ddelta} immediately gives for any $D\in D_{\lambda}$ that \[\fair (D,\MoLongs (D))=O\left(\frac{m}{n\eps}\ln \frac{n\lambda}{\beta}+m\lambda\right)\] with probability at least $1-\beta$.
\end{proof}

Comparing against the results for $\ctm$, the new $\ln (n\lambda)$ term in the bound is relatively inconsequential but the $m\lambda$ term is significant. This additional error stems from the $\spm_\lambda$ datasets's added freedom. Since they can stray further from singlepeakedness, neighboring datasets can now differ more substantially on $\fair$ (in particular, by roughly $m\lambda$). Finally, like in the case of $\ctm$, we can build off Theorems~\ref{p_bound_mp_ddelta} and \ref{fair_bound_mp_ddelta} to obtain a bound on $\swdiff$.

    \begin{theorem}[High probability bound on $\swdiff$ for $\MoLongs$ over $\spm_\lambda$]\label{swdiff_bound_ddelta}
    There exists a universal constant $C$ such that for every $\beta \in (0, 1/3), \eps \in (0,1)$, and $n\in \mathbb{N}$ satisfying $n\eps \ge C\cdot \ln (1/\beta)$ and $n\ge \beta/\lambda$, it holds for every $D\in \spm_\lambda$ that \[\swdiff (D,\MoLongs (D))=O\left(\frac{m}{n\eps^2}\left(\ln \frac{n\lambda}{\beta}\right)^2+\frac{m\lambda}{\eps}\ln \frac{n\lambda}{\beta}\right)\] with probability at least $1-\beta$ for widening parameter $\alpha^* = 1/(n\eps)$.
\end{theorem}

\begin{proof}
    This proof uses the same technique as the one for Theorem~\ref{swdiff_bound_mp} so we provide a briefer sketch. In particular, we again note that \[\swdiff(D,\ell) \le (2p_{\alpha}(D,\ell)-1)\cdot \fair (D,\ell).\] Like in Theorem~\ref{fair_bound_mp_ddelta}, we set $\alpha = 1/(n\eps)$. Combining the high probability bounds for $p_{\alpha^*}$ and $\fair$ from Corollary~\ref{p_bound_mp_ddelta_with_alpha} and Theorem~\ref{fair_bound_mp_ddelta}, we have that $\swdiff(D,\MoLongs(D))$ is at most  \[O\left(\frac{1}{\eps} \ln \frac{n\lambda}{\beta} \right)\cdot O\left(\frac{m}{n\eps}\max \frac{n\lambda}{\beta} + m\lambda \right)=O\left(\frac{m}{n\eps^2}\left(\ln \frac{n\lambda}{\beta}\right)^2+\frac{m\lambda}{\eps}\max\ln \frac{n\lambda}{\beta}\right)\] as required.
\end{proof}

\subsection{Information-theoretic lower bounds on loss}\label{lower_bounds}

Next, we benchmark the strength of our upper bounds by deriving counterpart information-theoretic lower bounds on $\fair$ and $\swdiff$ that every DP mechanism must incur. We will construct such lower bounds by using a more general version of the idea from Theorem~\ref{impossibility_dp_fair}.

\begin{theorem}[Direct Lower Bound]\label{dir_lower_bound}
    Let $D_1, D_2\in V^n$ be two distinct datasets and let $Y_1, Y_2$ be disjoint subsets of the outcome space $\mathcal{Y}$. For every $\eps$-DP mechanism $\M: V^n\to \mathcal{Y}$, it must hold that \[\min_{i\in\{1,2\}} \Pr[\M(D_i)\in Y_i]\le \frac{e^{d_{co}(D_1,D_2)\eps}}{e^{d_{co}(D_1,D_2)\eps} + 1}\] where $d_{co}(D_1,D_2)$ is the change-one distance between the two datasets.
\end{theorem}

The sets $Y_i$ here represent the range of outputs that are ``good'' (as measured by amount of loss with respect to a loss metric) for a particular dataset $D_i$. Therefore, for two datasets with different ``good output ranges,'' Theorem~\ref{dir_lower_bound} demonstrates that the similarity in their output distributions under DP will limit the likelihood that the actual output falls into the desirable ranges.

\begin{proof}
    By group privacy, it holds that \[\Pr[\M(D_1)\in Y_1] \le e^{d_{co}(D_1,D_2)\eps}\cdot  \Pr[\M(D_2)\in Y_1]\] which means $\Pr[\M(D_2)\not\in Y_2]\ge \Pr[\M(D_2)\in Y_1]\ge e^{-d_{co}(D_1,D_2)\eps} \cdot \Pr[\M(D_1)\in Y_1]$ and therefore \[\Pr[\M(D_2)\in Y_2]\le 1-e^{-d_{co}(D_1,D_2)\eps}\cdot \Pr[\M(D_1)\in Y_1]\] which means that \begin{align*}
        \min_i \Pr[\M(D_i)\in Y_i] &\le \min (\Pr[\M(D_1)\in Y_1], 1-e^{-d_{co}(D_1,D_2)\eps}\Pr[\M(D_1)\in Y_1])\\
        &\le \frac{1}{1+e^{-d_{co}(D_1,D_2)\eps}}\\
        &= \frac{e^{d_{co}(D_1,D_2)\eps}}{1+e^{d_{co}(D_1,D_2)\eps}}.
    \end{align*}
\end{proof}

As before, we begin by finding bounds over $\ctm$.

\subsubsection{Lower bounds over $\ctm$}

The key idea is to find datasets in $\ctm$ that are close in change-one distance relative to how far apart their medians are.

    \begin{theorem}[Lower Bound for $\fair$ over $\ctm$]\label{lower_bound_fair}
        There exist $\beta>0$ and $C$ such that for every $n\ge 5$, every $\eps\in (0,1)$ satisfying $n\eps \ge C$, and every $\eps$-DP mechanism $\M$, it holds for some $D\in \ctm$ that \[\fair (D,\M(D)) = \Omega \left(\frac{m}{n\eps}\ln \frac 1\beta\right)\] with probability at least $\beta$.
    \end{theorem}

\begin{proof}
    Let $D_0$ be the dataset which evenly spaces $n$ agents over $V$ so that $x_1 = -m/2$ and $x_n = m/2$, with $|x_{i+1}-x_i| = m/(n-1)$ for all $i\in \{1,\ldots, n-1\}$. Meanwhile, consider values of $\gamma <m/3$ that are integer multiples of $m/(n-1)$, noting that at least one such $\gamma$ always exists when $n\ge 5$. For any fixed such $\gamma$, let $D_\gamma$ be the dataset with $n_\gamma = 1+2\gamma(n-1)/m$ agents stacked at $\T(D_0) - \gamma$ and one agent at each of \[\T(D_0) -\gamma +j\frac{m}{n-1}\] for all \[j\in \left\{-\frac{n-n_\gamma}{2},-\frac{n-n_\gamma}{2}+1,\ldots, \frac{n-n_\gamma}{2}-1, \frac{n-n_\gamma}{2}\right\}-\{0\}.\]

\begin{figure}[htp]
\begin{subfigure}{\linewidth}
\centering
        \begin{tikzpicture}

        \draw[thick] (-6,0) -- (6,0) node[midway, above] {};
        
        \draw[solid, thick] (-6,-0.2) -- (-6,0.2) node[below=10] {$-m/2$};
        \draw[solid, thick] (-3,-0.1) -- (-3,0.1) node[below=7] {$-m/4$};
        \draw[solid, thick] (0,-0.1) -- (0,0.1) node[below=7] {$\vphantom{/}0$};
        \draw[solid, thick] (3,-0.1) -- (3,0.1) node[below=7] {$m/4$};
        \draw[solid, thick] (6,-0.2) -- (6,0.2) node[below=10] {$m/2$};
        
        \foreach \x in {-6,-5,-4,-3,-2,-1,1,2,3,4,5,6} {
            \filldraw[blue] (\x, 0.5) circle (2pt); 
        }

        \node at (0, 0.5) {\textbf{$\cdots$}};
        
        \node[blue] at (0, -1.6) {$\mathcal{T}(D_0)$};
        \draw[blue, thick,->] (0,-1.3) -- (0,-0.9);
    \end{tikzpicture}
\caption{Depiction of Dataset $D_0$. Blue points indicate agent locations along $V$.}
\end{subfigure}
\vspace{0.5cm}

\begin{subfigure}{\linewidth}
\centering
    \begin{tikzpicture}
 
        \draw[thick] (-6,0) -- (6,0) node[midway, above] {};
        
        \draw[solid, thick] (-6,-0.2) -- (-6,0.2) node[below=10] {$-m/2$};
        \draw[solid, thick] (-3,-0.1) -- (-3,0.1) node[below=7] {$-m/4$};
        \draw[solid, thick] (0,-0.1) -- (0,0.1) node[below=7] {$\vphantom{/}0$};
        \draw[solid, thick] (3,-0.1) -- (3,0.1) node[below=7] {$m/4$};
        \draw[solid, thick] (6,-0.2) -- (6,0.2) node[below=10] {$m/2$};
        \draw[solid, thick] (-2,-0.1) -- (-2,0.1) node[below=4mm] {};
        \foreach \y in {0.5,0.8} {
            \filldraw[red] (-2,\y) circle (2pt);
        }
        \node at (-2,1.35) {\vdots};
        \foreach \y in {1.7} {
            \filldraw[red] (-2,\y) circle (2pt);
        }
        
        \foreach \x in {-6,-5,-4,0,1,2} {
            \filldraw[red] (\x, 0.5) circle (2pt);  
        }

        \node at (-1, 0.5) {\textbf{$\cdots$}};
        \node at (-3, 0.5) {\textbf{$\cdots$}};
        
        \node[red] at (-1.65, -1.4) {$\mathcal{T}(D_\gamma) = \T(D_0) -\gamma$};
        \draw[red, thick,->] (-2,-1.1) -- (-2,-0.5);

        \draw [thick, decorate,decoration={brace,mirror,amplitude=5pt}] (-2.4,1.8) -- (-2.4,0.4) node[midway,xshift=-16pt] {$n_\gamma$};
    \end{tikzpicture}
\caption{Depiction of Dataset $D_\gamma$. Red points indicate agent locations along $V$.}
\end{subfigure}        \caption{The datasets, $D_0$ and $D_\gamma$, used in the proof of Theorem~\ref{lower_bound_fair}.}
\end{figure}
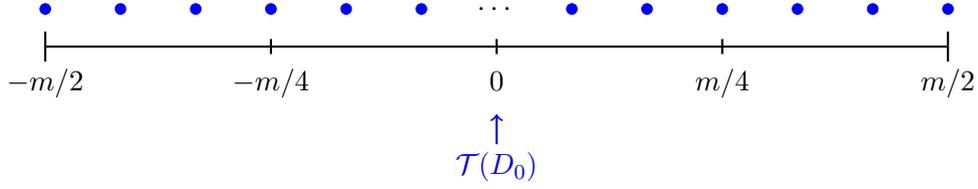
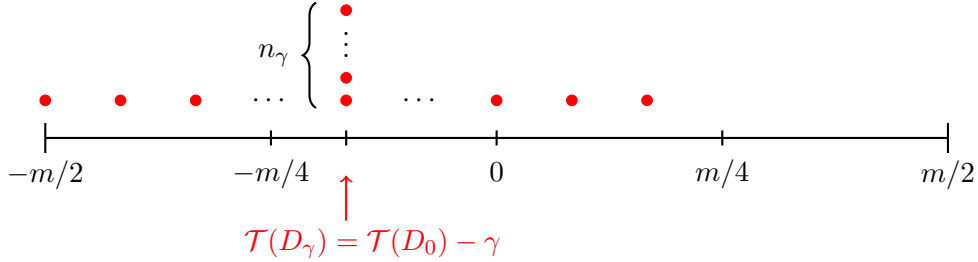
    
    Note that $D_\gamma$ is symmetric about $\T(D_0)-\gamma$ so $\T(D_\gamma) = \T(D_0)-\gamma$. Additionally, since $D_\gamma$ is peaked at its median and otherwise uniform, and $D_0$ is fully uniform, both datasets lie in $\ctm$ as needed. Moreover, because $\gamma$ is a multiple of $m/(n-1)$, many agent locations overlap between $D_0$ and $D_\gamma$. In particular, one can construct $D_\gamma$ from $D_0$ by simply moving the rightmost $n_\gamma -1 = 2\gamma(n-1)/m$ agents in $D_0$ and putting them all at $\T(D_\gamma)$ in $D_\gamma$, as visualized in Figure~\ref{fig:transform_datasets}.

    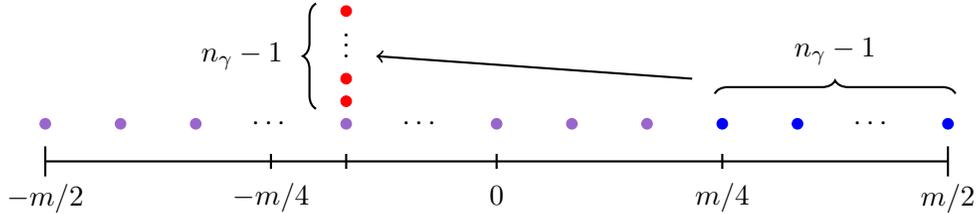
\begin{figure}[h]
    \centering
\begin{tikzpicture}
  
    \draw[thick] (-6,0) -- (6,0) node[midway, above] {};
    
    \draw[solid, thick] (-6,-0.2) -- (-6,0.2) node[below=10] {$-m/2$};
    \draw[solid, thick] (-3,-0.1) -- (-3,0.1) node[below=7] {$-m/4$};
    \draw[solid, thick] (0,-0.1) -- (0,0.1) node[below=7] {$\vphantom{/}0$};
    \draw[solid, thick] (3,-0.1) -- (3,0.1) node[below=7] {$m/4$};
    \draw[solid, thick] (6,-0.2) -- (6,0.2) node[below=10] {$m/2$};
    \draw[solid, thick] (-2,-0.1) -- (-2,0.1) node[below=4mm] {};
    
    \foreach \y in {0.8, 1.1} {
        \filldraw[red] (-2,\y) circle (2pt);
    }
    \node at (-2,1.65) {\vdots};
    \foreach \y in {2.0} {
        \filldraw[red] (-2,\y) circle (2pt);
    }
    
    \foreach \x in {-6,-5,-4,-2,0,1,2} {
        \filldraw[amethyst] (\x, 0.5) circle (2pt);  
    }

    \foreach \x in {3,4,6} {
        \filldraw[blue] (\x, 0.5) circle (2pt);  
    }

    \node at (-1, 0.5) {\textbf{$\cdots$}};
    \node at (-3, 0.5) {\textbf{$\cdots$}};
    \node at (5, 0.5) {\textbf{$\cdots$}};

    \draw [thick, decorate, decoration={brace,mirror,amplitude=5pt}] (-2.4,2.1) -- (-2.4,0.7) node[midway,xshift=-28pt] {$n_\gamma-1$};

    \draw [thick, decorate, decoration={brace, amplitude=5pt}] (2.9, 0.9) -- (6.1, 0.9) node[midway, yshift=16pt] {$n_{\gamma}-1$};

    \draw[->, black, thick] (2.6, 1.1) -- (-1.6, 1.4);
\end{tikzpicture}
\caption{A depiction of how to transform $D_0$ into $D_\gamma$ by moving $n_\gamma-1$ agents. The purple points denote agent locations shared across $D_0$ and $D_\gamma$, whereas the blue points denote agent locations in $D_0$ that get moved to the red points in $D_\gamma$.}
\label{fig:transform_datasets}
    \end{figure}

    Therefore, the change-one distance between $D_0$ and $D_\gamma$ is $h=\gamma(n-1)/m$. Let $Y_0 = (\T(D_0) -\gamma/2, \T(D_0)+\gamma/2)$ and $Y_\gamma = (\T(D_\gamma)-\gamma/2, \T(D_\gamma)+\gamma/2)$ be intervals along $V$. Note that $Y_\gamma$'s left endpoint is well-defined by the condition $\gamma<m/3$. Then, by Theorem~\ref{dir_lower_bound}, \[\min_{i\in \{0,\gamma\}} \Pr[\M(D_i)\in Y_i]\le \frac{e^{h\eps}}{1+e^{h\eps}}\] which means \[\max_{i\in \{0,\gamma\}} \Pr[\fair (D_i,\M(D_i))> \gamma/2]=\max_{i\in \{0,\gamma\}} \Pr[\M(D_i)\not\in Y_i]> \frac{1}{1+e^{h\eps}}\ge \frac{1}{2e^{h\eps}}.\] There then exists a \[\gamma = \Theta \left(\frac{m}{n\eps}\ln \frac 1\beta\right)\]for which $\max_{i\in \{0,\gamma\}} \Pr[\fair (D_i,\M(D_i))> \gamma/2]$ exceeds $\beta$, implying the result for $n$ and $\eps$ where $n\eps $ is sufficiently large (so that $\gamma < m/3$).
\end{proof}
We can use a similar technique to find a lower bound for $\swdiff$ over $\ctm$. In addition, we leverage the relationship between $\fair$ and $\swdiff$ so that Theorem~\ref{lower_bound_fair} also has implications on $\swdiff$.

    \begin{theorem}[Lower Bound for $\swdiff$ over $\ctm$]\label{lower_bound_swdiff}
        There exist $\beta>0$ and $C$ such that for every $n\ge 5$, every $\eps\in (0,1)$, and every $\eps$-DP mechanism $\M$, it holds for some $D\in \ctm$ that \[\swdiff (D,\M(D))=\Omega \left(\frac{m}{n\eps ^2}\left(\ln \frac 1\beta\right)^2\right)\] with probability at least $\beta$.
    \end{theorem}
    
\begin{proof}
    Define datasets $D_0$ and $D_\gamma$ in the same way as in the proof for Theorem~\ref{lower_bound_fair}. Utilizing the closed form for $\swdiff$ given in Theorem~\ref{swdiff_closed_form} and the structure of $D_0$ and $D_\gamma$, we can write $\swdiff (D,y)$ as a function of $d(\T(D), y)=\fair (D,y)$ for these two datasets.

    First, consider $D_0$ and suppose $\fair (D_0,\M(D_0))=r$ where $r$ is an integer multiple of $m/(n-1)$. Assume without loss of generality that $\M(D_0)< \T(D_0)$. Then, the set of crossed agents (from Definition~\ref{crossed_agents}) corresponds to the agents located at \[\T(D_0) - j\frac{m}{n-1}\] for $j=\{1,2,\ldots, r(n-1)/m\}$. Therefore, by Theorem~\ref{swdiff_closed_form}, we have that \begin{align*}
        \swdiff (D_0, \M(D_0)) &= r+2\sum_{j=1}^{\frac{r(n-1)}{m}} \left(j\cdot \frac{m}{n-1}\right)\\
        &= \Theta \left(\sum_{j=1}^{\frac{r(n-1)}{m}} \left(j\cdot \frac{m}{n-1}\right)\right)= \Theta \left(\frac{r^2n}{m}\right).
    \end{align*} Now, consider $D_\gamma$ and again suppose $\fair (D_\gamma,\M(D_\gamma))=r$ for $r$ that is an integer multiple of $m/(n-1)$. The set of crossed agents would now be comprised of two subsets of agents:
    \begin{itemize}
        \item All the agents located between $\M(D_\gamma)$ and $\T(D_\gamma)$ that are $jm/(n-1)$ away from $\T(D_\gamma)$ for $j=\{1,2,\ldots, r(n-1)/m\}$ (as in the case of $D_0$).
        \item Half of the non-median agents located at $\T(D_\gamma)$, of which there are $2\gamma(n-1)/m$ total.
    \end{itemize} Consequently, by Theorem~\ref{swdiff_closed_form} again, \[\swdiff (D_\gamma, \M(D_\gamma)) = \Theta\left(\frac{\gamma (n-1)}{m}r+\sum_{j=1}^{\frac{r(n-1)}{m}} \left(j\cdot \frac{m}{n-1}\right)\right) = \Theta \left(\frac{n}{m}r(r+\gamma)\right).\] When $r=\Omega (\gamma)$, then both $\swdiff(D_0,\M(D_0))$ and $\swdiff(D_\gamma,\M(D_\gamma))$ are $\Theta \left(nr^2/m\right)$. By Theorem~\ref{lower_bound_fair}, there is an at least $\beta$ probability that one of $\fair (D_0, \M(D_0))$ and $\fair (D_\gamma, \M(D_\gamma))$ is $\Omega (m\ln (1/\beta)/(n\eps))$ when $n\eps$ is assumed to be sufficiently large, which translates to an at least $\beta$ probability that one of $\swdiff(D_0,\M(D_0))$ and $\swdiff(D_\gamma,\M(D_\gamma))$ is \[\Omega \left(\frac nm \left(\frac{m}{n\eps}\ln \frac 1\beta\right)^2\right)=\Omega\left(\frac{m}{n\eps^2}\left(\ln \frac 1\beta\right)^2\right), \] as claimed.
\end{proof}

A comparison to the corresponding upper bounds for $\ctm$ in Theorems~\ref{fair_bound_mp} and \ref{swdiff_bound_mp} reveals that our upper and lower bounds match, so $\MoLong$ is fully optimal on $\swdiff$ and $\fair$ under datasets from $\ctm$.

\subsubsection{Lower bounds over $\spm_\lambda$}

Recall by Theorem~\ref{dc_is_in_ddelta} that $\ctm\subseteq \spm_\lambda$ for any $\lambda$ where $\spm_\lambda$ is nontrivial. Therefore, we can directly transfer the lower bounds for $\ctm$ as starting points.

\begin{theorem}[Theorems~\ref{lower_bound_fair} and \ref{lower_bound_swdiff} re-stated for $\spm_\lambda$]
There exist $\beta>0$ and $C$ such that for every $n\ge 5$, every $\eps\in (0,1)$ satisfying $n\eps \ge C$, and every $\eps$-DP mechanism $\M$, it holds for some $D_f, D_s\in \ctm \subset \spm_\lambda$ that
\begin{itemize}
    \item $\fair (D_f, \M(D_f))=\Omega \left(\frac{m}{n\eps}\ln \frac 1\beta\right)$ with probability at least $\beta$.
    \item $\swdiff (D_s, \M(D_s))=\Omega \left(\frac{m}{n\eps^2}\left(\ln \frac 1\beta\right)^2\right)$ with probability at least $\beta$.
\end{itemize}
\end{theorem}

If we compare these bounds to the upper bounds that $\MoLong$ attains on $\spm_\lambda$ in Theorems~\ref{fair_bound_mp_ddelta} and \ref{swdiff_bound_ddelta}, we see that the primary difference is the $m\lambda$ term. However, we indeed find the $m\lambda$ component is an unavoidable part of the loss.

\begin{theorem}[Another Lower Bound for $\fair$ over $\spm_\lambda$]\label{lower_bound_ddelta_mdelt}
    For every $n, \eps$, and $\eps$-DP mechanism $\M$, it holds for some $D\in \spm_\lambda$ for which \[\fair (D,\M(D))=\Omega (m\lambda)\] with probability at least $e^{-\eps}$.
\end{theorem}

\begin{proof}
    Define datasets $D_0$ as the dataset with one agent at each of $-m/2+jm/n$ for every $j\in \{1,2,\ldots, n\}$, so agents are evenly spaced $m/n$ apart. Now, let $s=\lfloor \lambda n\rfloor -1$. We now define two datasets $D_1$ and $D_2$ that are transformations of $D_0$. The goal is to make $D_1$ and $D_2$ neighboring datasets that are both as sparse as possible (under the constraints of $\spm_\lambda$) on one side of their median.
    
    The following process is illustrated by Figure~\ref{fig:transform_datasets_d1d2}. Construct $D_1$ from $D_0$ by moving the $s$ agents originally (in $D_0$) at $-m/2 + jm/n$ for $j\in \{\lceil n/2\rceil -s+1, \ldots, \lceil n/2\rceil\}$ all to $-m/2$, and leaving the remaining agents in their original locations. Meanwhile, construct $D_2$ from $D_0$ very similarly by moving the $s$ agents originally at $-m/2 + jm/n$ for $j\in \{\lceil n/2\rceil -s, \ldots, \lceil n/2\rceil-1\}$ all to $-m/2$ and leaving the rest be. One may verify directly that both $D_1,D_2\in \spm_\lambda$ since they are within $\lambda$ K-S distance from the uniform distribution $P\in \mathcal{P}$ over $[-m/2, m/2]$. The final datasets $D_1$ and $D_2$ are depicted in Figure~\ref{fig:datasets_d1d2}.
    
\begin{figure}[h]
\begin{subfigure}{\linewidth}
    \centering
\begin{tikzpicture}
 
    \draw[thick] (-6,0) -- (6,0) node[midway, above] {};
    \draw[solid, thick] (-6,-0.2) -- (-6,0.2) node[below=10] {$-m/2$};
    \draw[solid, thick] (-3,-0.1) -- (-3,0.1) node[below=7] {$-m/4$};
    \draw[solid, thick] (0,-0.1) -- (0,0.1) node[below=7] {$\vphantom{/}0$};
    \draw[solid, thick] (3,-0.1) -- (3,0.1) node[below=7] {$m/4$};
    \draw[solid, thick] (6,-0.2) -- (6,0.2) node[below=10] {$m/2$};
    \draw[solid, thick] (0.75,-0.1) -- (0.75,0.1) node[below=4mm] {};
    \draw[solid, thick] (-2.25,-0.1) -- (-2.25,0.1) node[below=4mm] {};

    \foreach \y in {0.5, 0.8} {
        \filldraw[red] (-6,\y) circle (2pt);
    }
    \node at (-6,1.35) {\vdots};
    \foreach \y in {1.7} {
        \filldraw[red] (-6,\y) circle (2pt);
    }

    \foreach \x in {-5.25,-4.50,-3.75,-3.00, -2.25, 1.50, 2.25, 3.00, 3.75, 4.50, 5.25, 6.00} {
        \filldraw[amethyst] (\x, 0.5) circle (2pt); 
    }

    \foreach \x in {0,-1.50, 0.75} {
        \filldraw[blue] (\x, 0.5) circle (2pt);  
    }

    \node at (-0.75, 0.5) {\textbf{$\cdots$}};

    \draw [thick, decorate, decoration={brace,mirror,amplitude=5pt}] (-6.4,1.8) -- (-6.4,0.4) node[midway,xshift=-16pt] {$s$};

    \draw [thick, decorate, decoration={brace, amplitude=5pt}] (-1.5, 0.9) -- (0.85, 0.9) node[midway, yshift=16pt] {$s$};
  
    \draw[->, black, thick] (-1.9, 1.1) -- (-5.6, 1.1);

    \node[blue] at (0.75, -1.5) {$\mathcal{T}(D_0) $};
    \draw[blue, thick,->] (0.75,-1.1) -- (0.75,-0.5);

    \node[red] at (-2.25, -1.5) {$\mathcal{T}(D_1) $};
    \draw[red, thick,->] (-2.25,-1.1) -- (-2.25,-0.5);
\end{tikzpicture}
\caption{Construction of $D_1$ from $D_0$}
\end{subfigure}
\vspace{0.5cm}

\begin{subfigure}{\linewidth}
\centering
\begin{tikzpicture}

    \draw[thick] (-6,0) -- (6,0) node[midway, above] {};
    \draw[solid, thick] (-6,-0.2) -- (-6,0.2) node[below=10] {$-m/2$};
    \draw[solid, thick] (-3,-0.1) -- (-3,0.1) node[below=7] {$-m/4$};
    \draw[solid, thick] (0,-0.1) -- (0,0.1) node[below=7] {$\vphantom{/}0$};
    \draw[solid, thick] (3,-0.1) -- (3,0.1) node[below=7] {$m/4$};
    \draw[solid, thick] (6,-0.2) -- (6,0.2) node[below=10] {$m/2$};
    \draw[solid, thick] (0.75,-0.1) -- (0.75,0.1) node[below=4mm] {};
    
    \foreach \y in {0.5, 0.8} {
        \filldraw[red] (-6,\y) circle (2pt);
    }
    \node at (-6,1.35) {\vdots};
    \foreach \y in {1.7} {
        \filldraw[red] (-6,\y) circle (2pt);
    }
    
    \foreach \x in {-5.25,-4.50,-3.75,-3.00, 0.75, 1.50, 2.25, 3.00, 3.75, 4.50, 5.25, 6.00} {
        \filldraw[amethyst] (\x, 0.5) circle (2pt); 
    }

    \foreach \x in {0,-1.50, -2.25} {
        \filldraw[blue] (\x, 0.5) circle (2pt);  
    }
    \node at (-0.75, 0.5) {\textbf{$\cdots$}};
    \draw [thick, decorate, decoration={brace,mirror,amplitude=5pt}] (-6.4,1.8) -- (-6.4,0.4) node[midway,xshift=-16pt] {$s$};
    
    \draw [thick, decorate, decoration={brace, amplitude=5pt}] (-2.35, 0.9) -- (0.10, 0.9) node[midway, yshift=16pt] {$s$};
    \draw[->, black, thick] (-2.65, 1.1) -- (-5.6, 1.1);

    \node[black] at (0.75, -1.5) {$\textcolor{blue}{\mathcal{T}(D_0)} = \textcolor{red}{\mathcal{T}(D_2)} $};
    \draw[black, thick,->] (0.75,-1.1) -- (0.75,-0.5);

\end{tikzpicture}
\caption{Construction of $D_2$ from $D_0$.}
\end{subfigure}
\caption{A depiction of how to construct $D_1$ and $D_2$ from $D_0$ by moving $s$ points. The purple points denote agent locations that stay fixed. Meanwhile, the $s$ agents at the blue points are moved to the red points. Note that for $D_1$, the median changes from the original median in $D_0$, whereas for $D_2$, the median remains the same as before.}
\label{fig:transform_datasets_d1d2}
    \end{figure}
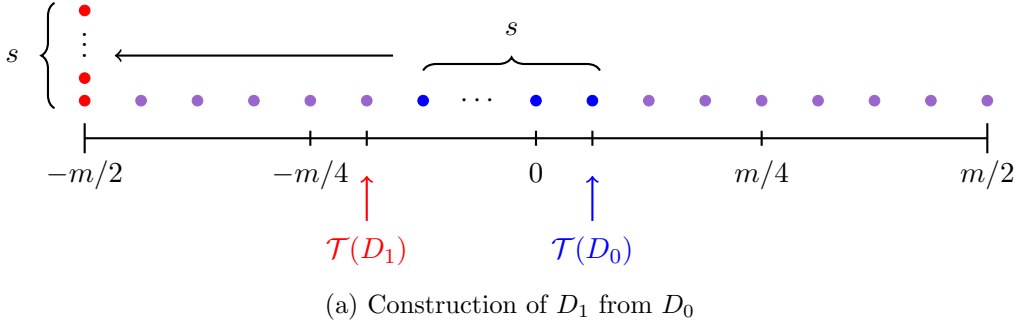
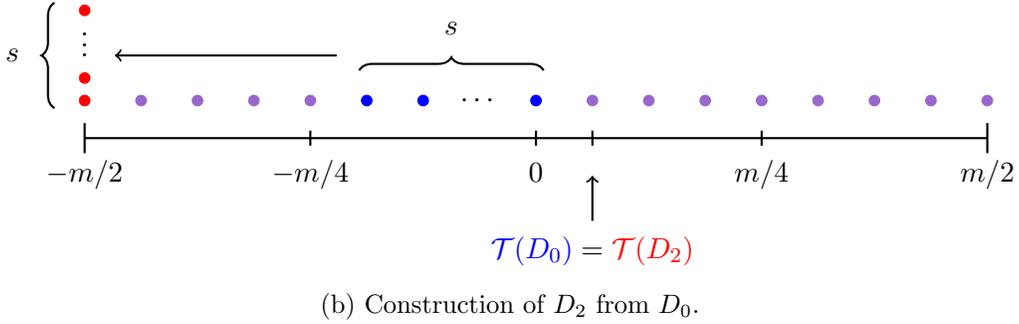

    \begin{figure}[h]
\begin{subfigure}{\linewidth}
    \centering
\begin{tikzpicture}
 
    \draw[thick] (-6,0) -- (6,0) node[midway, above] {};

    \draw[solid, thick] (-6,-0.2) -- (-6,0.2) node[below=10] {$-m/2$};
    \draw[solid, thick] (-3,-0.1) -- (-3,0.1) node[below=7] {$-m/4$};
    \draw[solid, thick] (0,-0.1) -- (0,0.1) node[below=7] {$\vphantom{/}0$};
    \draw[solid, thick] (3,-0.1) -- (3,0.1) node[below=7] {$m/4$};
    \draw[solid, thick] (6,-0.2) -- (6,0.2) node[below=10] {$m/2$};
    \draw[solid, thick] (-2.25,-0.1) -- (-2.25,0.1) node[below=4mm] {};

    \foreach \y in {0.5, 0.8} {
        \filldraw[red] (-6,\y) circle (2pt);
    }
    \node at (-6,1.35) {\vdots};
    \foreach \y in {1.7} {
        \filldraw[red] (-6,\y) circle (2pt);
    }
    
    \foreach \x in {-5.25,-4.50,-3.75,-3.00, -2.25, 1.50, 2.25, 3.00, 3.75, 4.50, 5.25, 6.00} {
        \filldraw[red] (\x, 0.5) circle (2pt);  
    }

    \draw [thick, decorate, decoration={brace,mirror,amplitude=5pt}] (-6.4,1.8) -- (-6.4,0.4) node[midway,xshift=-16pt] {$s$};

    \node[red] at (-2.25, -1.5) {$\mathcal{T}(D_1) $};
    \draw[red, thick,->] (-2.25,-1.1) -- (-2.25,-0.5);
\end{tikzpicture}
\caption{Depiction of Dataset $D_1$.}
\end{subfigure}
\vspace{0.5cm}

\begin{subfigure}{\linewidth}
\centering
\begin{tikzpicture}

    \draw[thick] (-6,0) -- (6,0) node[midway, above] {};

    \draw[solid, thick] (-6,-0.2) -- (-6,0.2) node[below=10] {$-m/2$};
    \draw[solid, thick] (-3,-0.1) -- (-3,0.1) node[below=7] {$-m/4$};
    \draw[solid, thick] (0,-0.1) -- (0,0.1) node[below=7] {$\vphantom{/}0$};
    \draw[solid, thick] (3,-0.1) -- (3,0.1) node[below=7] {$m/4$};
    \draw[solid, thick] (6,-0.2) -- (6,0.2) node[below=10] {$m/2$};
    \draw[solid, thick] (0.75,-0.1) -- (0.75,0.1) node[below=4mm] {};
    
    \foreach \y in {0.5, 0.8} {
        \filldraw[red] (-6,\y) circle (2pt);
    }
    \node at (-6,1.35) {\vdots};
    \foreach \y in {1.7} {
        \filldraw[red] (-6,\y) circle (2pt);
    }
    
    \foreach \x in {-5.25,-4.50,-3.75,-3.00, 0.75, 1.50, 2.25, 3.00, 3.75, 4.50, 5.25, 6.00} {
        \filldraw[red] (\x, 0.5) circle (2pt);  
    }
    \draw [thick, decorate, decoration={brace,mirror,amplitude=5pt}] (-6.4,1.8) -- (-6.4,0.4) node[midway,xshift=-16pt] {$s$};
    
    \node[red] at (0.75, -1.5) {$\mathcal{T}(D_2)$};
    \draw[red, thick,->] (0.75,-1.1) -- (0.75,-0.5);

\end{tikzpicture}
\caption{Depiction of Dataset $D_2$.}
\end{subfigure}
\caption{A comparison of the resulting datasets $D_1$ and $D_2$ in the proof of Theorem~\ref{lower_bound_ddelta_mdelt}. Observe that the two datasets share $n-1$ points and only differ in the median agent, with $|\T(D_1)-\T(D_2)| = s\frac mn$.}
\label{fig:datasets_d1d2}
    \end{figure}
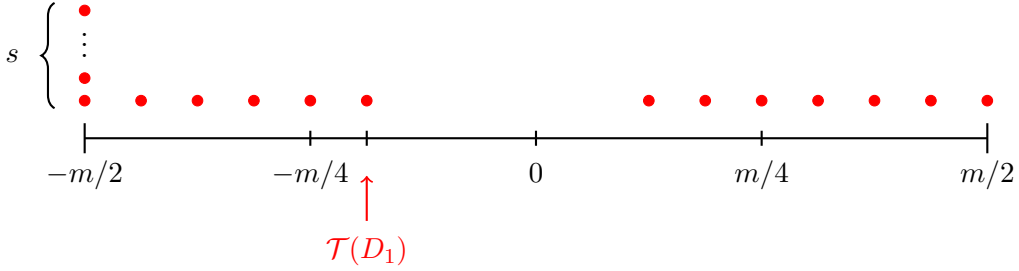
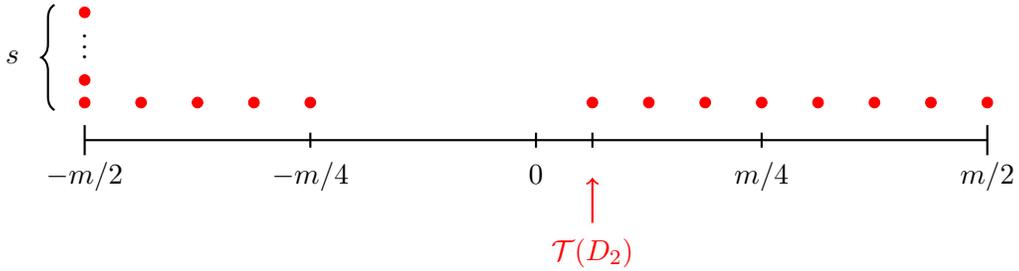

    Observe that $D_1$ and $D_2$ differ in the placement of only one agent. However, $\T(D_1) = -m/2+(\lceil n/2\rceil-s) \cdot m/n$ and $\T(D_2) = -m/2+\lceil n/2\rceil \cdot m/n$ so \[|\T(D_1)-\T(D_2)| = s\frac{m}{n}.\] Consider now the intervals $Y_1 = (\T(D_1)-sm/(2n), \T(D_1)+sm/(2n))$ and $Y_2 = (\T(D_2)-sm/(2n), \T(D_2)+sm/(2n))$. By Theorem~\ref{dir_lower_bound}, it follows that \[\max_{i\in \{1,2\}} \Pr\left[\fair (D_i,\M(D_i))\ge \frac{sm}{2n}\right]=\Pr[\M(D_i)\in Y_i] \le \frac{e^\eps}{e^\eps+1}\] and since $s=\Theta (n\lambda)$, this implies that for one of $i\in \{1,2\}$ there is at least $e^{-\eps}$ probability of $\fair (D_i,\M(D_i))$ being $\Omega (m\lambda)$.
\end{proof}

The intuition for this result is that we consider the maximum ``sparsity'' a dataset in $\spm_\lambda$ can have around its median. For instance, across the set of all possible datasets, we saw in the proof of Theorem~\ref{impossibility_dp_fair} that there exists a dataset that is incredibly sparse on one side of its median, with no agents within distance $m$ from it. Meanwhile, the proof of Theorem~\ref{lower_bound_ddelta_mdelt} illustrates that across $\spm_\lambda$, we can find a dataset that has no agents within distance $\Theta (m\lambda)$ from its median on one side.

Since it always holds that $\fair (D,\ell)\le \swdiff (D,\ell)$, we also get another lower bound on $\swdiff$ for free from Theorem~\ref{lower_bound_ddelta_mdelt}.

\begin{corollary}[Another Lower Bound for $\swdiff$ over $\spm_\lambda$]\label{lower_bound_swdiff_mdelt}
    For every $n, \eps$, and $\eps$-DP mechanism $\M$, it holds for some $D\in \spm_\lambda$ for which \[\swdiff (D,\M(D))=\Omega (m\lambda)\] with probability at least $e^{-\eps}$.
\end{corollary}

By combining our previous theorems, we now obtain the following key lower bounds over $\spm_\lambda$.

    \begin{theorem}[Lower Bound for $\fair$ over $\spm_\lambda$]\label{final_lower_bound_fair_ddelt}
        There exist $\beta>0$ and $C$ such that for every $n\ge 5$, every $\eps\in (0,1)$ satisfying $n\eps \ge C$, and every $\eps$-DP mechanism $\M$, it holds for some $D\in \spm_\lambda$ that \[\fair (D,\M(D))=\Omega \left(\frac{m}{n\eps}\ln \frac 1\beta+m\lambda\right)\] with probability at least $\beta$.
    \end{theorem}

    \begin{theorem}[Lower Bound for $\swdiff$ over $\spm_\lambda$]\label{final_lower_bound_swdiff_ddelt}
        There exist $\beta>0$ and $C$ such that for every $n\ge 5$, every $\eps\in (0,1)$ satisfying $n\eps \ge C$, and every $\eps$-DP mechanism $\M$, it holds for some $D\in \spm_\lambda$ that \[\swdiff (D,\M(D))=\Omega \left(\frac{m}{n\eps ^2}\left(\ln \frac 1\beta\right)^2+m\lambda\right)\] with probability at least $\beta$.
    \end{theorem}

\begin{proof}[Proof (both Theorem~\ref{final_lower_bound_fair_ddelt} and Theorem~\ref{final_lower_bound_swdiff_ddelt})]
    Since $\eps=O(1)$, an event that occurs with probability $e^{-\eps}$ must occur with at least constant probability $\beta$ for some $\beta$. Therefore, Theorems~\ref{lower_bound_ddelta_mdelt} and \ref{lower_bound_swdiff_mdelt} fit directly into the framework of our usual lower bounds.

    Merging Theorems~\ref{lower_bound_fair} and \ref{lower_bound_ddelta_mdelt} gives the lower bound of \[\max\left(\Omega \left(\frac{m}{n\eps}\right), \Omega(m\lambda)\right)= \Omega \left(\frac{m}{n\eps} + m\lambda\right),\] for $\fair$, and merging Theorems~\ref{lower_bound_swdiff} and \ref{lower_bound_swdiff_mdelt} similarly gives the lower bound for $\swdiff$.
\end{proof}

\section{Discussion}

We can compare our results across Sections~\ref{upper_bounds} and \ref{lower_bounds} via the following Table~\ref{tab:our_contributions_final}, which replicates Table~\ref{tab:our_contributions} from the introduction.

\begin{table}[h]
    \centering
    \small
    \renewcommand{\arraystretch}{1.3}
    \begin{adjustbox}{max width=\textwidth}
    \begin{tabular}{|c|c||c|c|}
    \hline
    \textbf{Dataset family} &  & \textbf{Lower bound} & \textbf{Upper bound achieved by $\MoLongs$} \\
    \hline\hline
    \multirow{2}{*}{\textbf{Smaller family $\ctm$}} 
    & $\fair$ 
    & $\Omega\left(\frac{m}{n\epsilon}\ln \frac{1}{\beta}\right)$ 
    & $O\left(\frac{m}{n\epsilon}\ln \frac{1}{\beta}\right)$ \\
    \cline{2-4}
    & $\swdiff$ 
    & $\Omega\left(\frac{m}{n\epsilon^2} \left(\ln \frac{1}{\beta} \right)^2 \right)$ 
    & $O\left(\frac{m}{n\epsilon^2} \left(\ln \frac{1}{\beta} \right)^2 \right)$ \\
    \hline
    \multirow{2}{*}{\textbf{Larger family $\spm_\lambda$}} 
    & $\fair$ 
    & $\Omega\left(\frac{m}{n\epsilon}\ln \frac{1}{\beta} + m\lambda\right)$ 
    & $O\left(\frac{m}{n\epsilon} \ln \frac{n\lambda}{\beta} + m\lambda\right)$ \\
    \cline{2-4}
    & $\swdiff$ 
    & $\Omega\left(\frac{m}{n\epsilon^2} \left(\ln \frac{1}{\beta} \right)^2 + m\lambda\right)$ 
    & $O\left(\frac{m}{n\epsilon^2} \left(\ln \frac{n\lambda}{\beta}\right)^2 + \frac{m\lambda}{\epsilon}\right)$ \\
    \hline
    \multirow{2}{*}{\textbf{General datasets}} 
    & $\fair$ 
    & $m/2$ 
    & -- \\
    \cline{2-4}
    & $\swdiff$ 
    & -- 
    & -- \\
    \hline
    \end{tabular}
    \end{adjustbox}
    \caption{Lower and upper bounds on loss derived in Sections~\ref{upper_bounds} and \ref{lower_bounds}.}
    \label{tab:our_contributions_final}
\end{table}

The lower bounds for $\ctm$ in Theorems~\ref{lower_bound_fair} and \ref{lower_bound_swdiff} match exactly with their counterparts in Theorems~\ref{fair_bound_mp} and \ref{swdiff_bound_mp}, so $\MoLong$ is exactly optimal on both $\fair$ and $\swdiff$ for $\ctm$. Meanwhile, the lower bounds on $\spm_\lambda$ also match closely to the upper bounds from Theorems~\ref{fair_bound_mp_ddelta} and \ref{swdiff_bound_ddelta}. The only discrepancies are the presence of additional $\max(1,\ln (n\lambda))$ factors in the upper bounds and an extra $1/\eps$ factor on $m\lambda$ in the $\swdiff$ upper bound. However, these extra terms are small enough that we would consider them largely unsubstantial. Indeed, under the common assumptions that $\epsilon=\Theta(1)$ and $\lambda = O(1/\sqrt{n})$, the $\max$ is $O(\ln n)$ and thus logarithmic relative to the factor of $n$ in the overall term's denominator, and the $1/\eps$ factor is $\Theta(1)$. Both differences are minor enough that we consider the differentially private mechanism $\MoLong$ to be virtually optimal on both $\fair$ and $\swdiff$ for $\spm_\lambda$.

Thus, $\MoLong$'s performance relative to the lower bounds on $\ctm$ and $\spm_\lambda$ suggests that on each of our two families, $\fair$ and $\swdiff$ can be simultaneously optimized over DP mechanisms. Therefore, while there is a tradeoff between privacy and each of social welfare and fairness in facility location mechanism design, there is no additional tradeoff when we consider all three objectives simultaneously, provided that the population data is sufficiently `natural.

There are a few directions for future work. First, one could focus on reconciling the small mismatches in bounds that still remain. Additionally, it may be worth analyzing this same tradeoff under different distributional assumptions used for similar relaxations in the DP median-finding literature~\cite{DL2009,NRS2011, BCS2018a, BCS2018b, ASSU2023, TVZ2020}. While we offer motivation for the families $\ctm$ and $\spm_\lambda$, one weakness of our conditions is that the single-peakedness property may still be too strong and not strictly necessary for a DP facility location mechanism to achieve reasonable utility. Moreover, it would be insightful to expand this framework to more general facility location settings in higher dimensional metric spaces, as the social welfare and fairness scoring functions are likely less well-behaved in those more complex settings. Lastly, it could be interesting to consider how the performance and fairness guarantees change under looser forms of differential privacy, such as approximate DP or other variants.

\printbibliography
\end{document}